%% file: DoublyWarpedBoundary2017.tex
\documentclass[a4paper,reqno,11pt]{article}
\usepackage[a4paper,bindingoffset=0.2in,%
left=1in,right=1in,top=1in,bottom=1in,%
           footskip=.25in]{geometry}
\usepackage{graphicx}

\usepackage{authblk}
\usepackage{amssymb}
\usepackage{amsthm}
\usepackage{amsmath}
\usepackage{color}
\usepackage[utf8]{inputenc}
\usepackage{tikz}
\usepackage[all]{xy}
\usepackage[bookmarksnumbered,colorlinks]{hyperref}
\usepackage{dsfont}
\usepackage[normalem]{ulem}

\usepackage{float} 

\usepackage[euler-digits]{eulervm}
\usepackage{enumerate}
\def\br#1\er{\textcolor{red}{#1}} %
\newcommand{\be}{\begin{equation}}
\newcommand{\ee}{\end{equation}}

\newcommand{\N}{{\mathbb N}}

\newcommand{\R}{{\mathbb R}}

\newcommand{\C}{\mathfrak{c}}
\newcommand{\J}{{\cal J}}

\newcommand{\ben}{\begin{enumerate}}
\newcommand{\een}{\end{enumerate}}
\newcommand{\bit}{\begin{itemize}}
\newcommand{\eit}{\end{itemize}}
\newcommand{\edoc}{\end{document}}

\newcommand{\cambios}[1]{{\color{black} #1}}
\newcommand{\ncambios}[1]{{\color{black} #1}}

\newcommand{\multiwarped}{doubly warped~}

\newcommand{\Integral}[5]{\int_{#1}^{#2} \frac{\sqrt{#3}}{\alpha_{#4}(s)}\left(\sum_{k=1}^{2} \frac{#5}{\alpha_{k}(s)} \right)^{-1/2}ds}

\newcommand{\point}[3]{(#1,#2,#3)}
\newcommand{\B}{b}

\title{Causality and c-completion of multiwarped spacetimes}
\date{}
\author[1]{Luis Alberto Aké}
\author[1]{José Luis Flores}
\author[2]{Jónatan Herrera}
\affil[1]{\small Departamento de Álgebra, Geometría y Topología, Facultad de Ciencias\\ Universidad de Málaga, Campus Teatinos, 29071 Málaga, Spain}
\affil[2]{\small Departamento de Matemáticas, Edificio Albert Einstein\\ Universidad de Córdoba, Campus de Rabanales, 14071 Córdoba, Spain}
\affil[ ]{\textit{Corresponding author: jonatanhf@gmail.com}}

\begin{document}
\newtheorem{thm}{Theorem}[section]
\newtheorem{prop}[thm]{Proposition}
\newtheorem{lemma}[thm]{Lemma}
\newtheorem{cor}[thm]{Corollary}
\newtheorem{conv}[thm]{Convention}
\theoremstyle{definition}
\newtheorem{defi}[thm]{Definition}
\newtheorem{notation}[thm]{Notation}
\newtheorem{exe}[thm]{Example}
\newtheorem{conj}[thm]{Conjecture}
\newtheorem{prob}[thm]{Problem}
\newtheorem{rem}[thm]{Remark}

\maketitle
\usetikzlibrary{matrix}

\begin{abstract}
In this paper a systematic study of the causal structure and global causality properties of multiwarped spacetimes is developed. This analysis is used to make a detailed description of the causal boundary of these spacetimes. Some applications of our results in examples of physical interest, for instance, in the context of Maldacena's conjecture, are considered.
%
%
\end{abstract}


\tableofcontents

\input{Introduction}

\input{Preliminaries}

\input{hierarchy}

\input{PartialBoundary}

\input{CompleteBoundary}

\section*{Acknowledgments}

The authors are partially supported by the Spanish Grant MTM2016-78807-C2-2-P (MINECO and FEDER funds). L. Aké also acknowledges a grant funded by the Consejo Nacional de Ciencia y Tecnolog\'ia (CONACyT), M\'exico.

\bibliographystyle{unsrt}
\bibliography{biblio2}

\end{document}

%% file: Introduction.tex
\section{Introduction}

The holographic principle \cite{tHooft:1993dmi,doi:10.1063/1.531249} states that the information of a particular space can be thought as encoded on a lower-dimensional boundary of the space, thus considering the original space as an hologram of the latter. One of the best understood examples of such a principle is the AdS/CFT correspondence, or Maldacena duality \cite{Mal}, where a dual description between the string theory on the bulk space (typically, the product of anti-de Sitter AdS$_n$ by a round sphere $\mathbb{S}^m$, or by another compact manifold) and a Quantum Field Theory without gravity on the  boundary of the initial space is achieved. Currently, there is a growing interest in the study of the realization of such a holographic principle with other bulk spaces \cite{PhysRevD.80.124008,GHODSI201079,1126-6708-2009-04-019}, particularly de Sitter spacetime dS$_{n}$ \cite{Witten:2001kn,1126-6708-2001-10-034,0264-9381-34-1-015009,Gibbons:1984kp}.

Two problems arise here. On the one hand, which {\em boundary} must we consider to formulate the holographic principle? In the original approach to the AdS/CFT correspondence, it is used the conformal boundary. However, this boundary presents important limitations generated by its {\em ad hoc} character: there is  no general formalism ensuring when a reasonably general spacetime has an intrinsic and unique conformal boundary. In fact, Bernstein, Maldacena and Nastase \cite{BMN} put forward different problems when the conformal boundary on plane waves is considered, and some years later, Marolf and Ross \cite{MR1} showed that, indeed, the conformal boundary is not available for non-conformally flat plane waves. This makes the alternative {\em causal boundary} a more suitable construction for the holographic principle, since it is intrinsic, conformally invariant and can be computed systematically.

On the other hand, anti-de Sitter spacetime is embedded in a string theory by making a warping product with a compact manifold. Due to the compactness of the latter, it is not difficult to obtain the causal boundary of the product from the causal boundary of AdS$_n$ (see for instance \cite{AF}). However, if de Sitter spacetime is considered, the no-go theorems (first due to Gibbons \cite{Gibbons:1984kp} and Maldacena, Nuñez \cite{doi:10.1142/S0217751X01003937}) ensure that there is no way to embed it in a string theory by a product with a compact manifold. There exist several ways to circumvent these no-go theorems, for instance, by considering warped product with non-compact Riemannian manifolds, but this complicates significatively the computation of the boundary.

\smallskip

 These problems motivate the systematic study of the causal boundary for the so-called {\em multiwarped spacetimes}, a class of spacetimes wide enough to cover the situations described above. A {\em multiwarped spacetime} $(V,g)$ can be written as $V=(a,b)\times M_1\times \dots \times M_n$, $-\infty\leq a<b\leq\infty$, and
\begin{equation}\label{eqqq}
g=-dt^2 + \alpha_1 g_1+\dots+\alpha_n g_n,
\end{equation}
where $\alpha_i:(a,b)\rightarrow\R$ are positive smooth functions and $(M_i,g_i)$ are Riemannian manifolds, for all $i=1,\ldots,n$.

As far as we know, the unique result in the literature about the causal boundary of these spacetimes is due to Harris \cite{H}:
\begin{thm}(Harris, 2008)\label{thm:harris}
Let $(V,g)$ be a multiwarped spacetime as above, and assume that (for some $c\in (a,b)$) the first $k$ warping functions, $1\leq i\leq k$, obey $\int_{c}^b(\alpha_i(s))^{-1/2}ds<\infty$, and the rest, $k+1\leq i\leq n$, obey $\int_{c}^b(\alpha_i(s))^{-1/2}ds=\infty$. Then the following hold:
  \begin{itemize}
  \item[(a)] If some Riemannian factor $M_i$ is incomplete, the future causal boundary $\hat{\partial} V$ has timelike-related elements.
  \item[(b)] If $M_i$ is incomplete for some $i\geq k+1$, the future causal boundary $\hat{\partial} V$ has null-related elements.
    \item[(c)] If neither of those occur, then $V$ has only spacelike future boundaries.
  \end{itemize}
In the last case, $\hat{\partial} V$ is homeomorphic to $M^0=M_1\times \dots \times M_k$. Furthermore, the future causal completion $\hat{V}$ is homeomorphic to $\left( (a,b]\times M^0\times M'\right)/\sim$, where $M'=M_{k+1}\times \dots \times M_n$ and $\sim$ is the equivalence relation defined by $(b,x^0,x')\sim (b,x^0,y')$ for any $x^0\in M^0$ and $x',y'\in M'$; $\hat{\partial} V$ appears there as $\{b\}\times M^0\times \{*\}$.
\end{thm}
\noindent This result covers the case of warped products of anti-de Sitter with compact manifold, however it does not provide a complete description of the boundary when the product of de Sitter with non-compact manifolds is considered.




 \smallskip

 The aim of this paper is twofold. First, we develop a systematic study of the causal structure and global causality properties of multiwarped spacetimes. Then, we use this approach to describe in full detail the causal boundary of these spacetimes by considering some mild integral hypothesis on the warping functions. Our main results for the future causal boundary, Thms. \ref{futurestructurefiniteconditions} and \ref{futurecomploneinfinite}, not only include the cases covered by Thm. \ref{thm:harris} for the future causal boundary, but also some additional ones. Concretely, we are able to remove the completeness condition on the Riemannian factors, and we also include the case where just one warping integral is infinite ($k+1=n$). Moreover, we consider the total c-boundary in Section \ref{sec:totalcompletion}, i.e. the construction obtained when the future and past causal boundaries are merged, concluding in Theorem \ref{thm:main}.
 Finally, we also discuss some relevant examples where our results are applicable.

The paper is organized as follows. In Section \ref{sec:preliminaries} we consider some preliminaries about the c-completion of spacetimes, focusing on the particular case of Robertson-Walker models, which we are going to use later. In Section \ref{sec:chronologicalrelation} we establish characterizations for the chronological and causal relations in doubly warped spacetimes. Then, in Section \ref{sec:causalladder}, we determine the position of these spacetimes into the causal ladder.
After that, in Sections \ref{sec:futurecompletion}, \ref{ss6} and \ref{sec:totalcompletion}, we use the machinery developed before to make a systematic study of the c-boundary of doubly warped spacetimes.
Finally, in Section \ref{sec:applications}, we discuss the applicability of our results by considering several examples of interest: Kasner models, intermediate Reissner--Nordstr\"om, and de Sitter models with general internal spaces.

\section{Preliminaries}
\label{sec:preliminaries}
\subsection{The c-completion of spacetimes}
The {\em causal completion} of spacetimes is a conformally invariant construction which consists of adding {\em
ideal points} to a strongly causal spacetime in such a way that any timelike
curve in the original spacetime acquires some endpoint in the new
space \cite{GKP}. The {\em c-completion}, which is the concrete formalization of the causal completion that we are going to adopt in this paper, requires some preliminary notions.

Let $(V,g)$ be a spacetime. We say that a non-empty
subset $P\subset V$ is a {\em past set} if it coincides
with its past; i.e. $P=I^{-}(P):=\{p\in V: p\ll q\;\hbox{for
some}\; q\in P\}$. The {\em common past} of $S\subset V$ is
defined by $\downarrow S:=I^{-}(\{p\in V:\;\; p\ll q\;\;\forall
q\in S\})$. From construction, the past and common past sets are
open. When a past set $P$ cannot be written as the union of two proper
past sets, we say that $P$
is an {\em indecomposable past} set, {\em IP}. The indecomposable past sets can be classified in two major classes. On the one hand, the IPs which
coincide with the past of some point of the spacetime,
$P=I^{-}(p)$, $p\in V$, are called {\em proper indecomposable past
  sets}, {\em PIP}. On the other hand, the IPs which are obtained as the past of inextendible timelike curve $\gamma$, $P=I^-(\gamma)$, are called {\em terminal indecomposable past sets}, {\em TIPs}. The dual
notions, {\em future set}, {\em common future}, {\em IF}, {\em
TIF} and {\em PIF}, are defined just by interchanging the roles of
past and future in previous definitions.

In order to construct the {\em future} and {\em past c-completions}, first
we have to identify each {\em event} $p\in V$ with its PIP,
$I^{-}(p)$, and PIF, $I^{+}(p)$. To achieve this, we need to restrict our attention on {\em
distinguishing} spacetimes. On the other hand, in order to obtain consistent topologies for the
c-completions, we need to focus on a somewhat more restrictive class
of spacetimes, the {\em strongly causal ones} (see Defn. \ref{ant}). These are
characterized by the fact that the PIPs and PIFs constitute a
sub-basis for the topology of the manifold $V$.

%

Once the events $p\in V$ have been identified with their corresponding PIPs, we define the {\em future c-boundary} $\hat{\partial}V$ of $V$
as the set of all the TIPs in $V$, and  {\em the future
c-completion} $\hat{V}$ as the set of all the IPs:
\[
V\equiv \hbox{PIPs},\qquad \hat{\partial}V\equiv
\hbox{TIPs},\qquad\hat{V}\equiv \hbox{IPs}.
\]
Analogously, each $p\in V$ can be identified with its corresponding PIF,
$I^+(p)$. The {\em past c-boundary} $\check{\partial}V$ of $V$ is
defined as the set of all the TIFs in $V$, and  {\em the past
c-completion} $\check{V}$ is the set of all the IFs:
\[
V\equiv \hbox{PIFs},\qquad \check{\partial}V\equiv
\hbox{TIFs},\qquad\check{V}\equiv \hbox{IFs}.
\]

In order to merge the future and past c-boundaries together to form the (total) c-boundary, the so-called S-relation comes into
play \cite{Sz}.
Let $\hat{V}_{\emptyset}:=\hat{V}\cup \{\emptyset\}$ (resp.
$\check{V}_{\emptyset}:=\check{V}\cup \{\emptyset\}$), and define the
S-relation $\sim_S$ in $\hat{V}_{\emptyset}\times
\check{V}_{\emptyset}$ as follows: First, in the case $(P,F)\in \hat{V}\times
\check{V}$, \be \label{eSz}  P\sim_S F \Longleftrightarrow \left\{
\begin{array}{l}
P \quad \hbox{is included and is a maximal IP into} \quad
\downarrow F
 \\
F \quad \hbox{is included and is a maximal IF into} \quad \uparrow
P.
\end{array} \right.
\end{equation}
Here, {\em maximal} means that no other $P'\in\hat{V}$ (resp.
$F'\in \check{V}$) satisfies the stated property and contains
strictly $P$ (resp. $F$). As it was proved by Szabados in
\cite{Sz}, $I^-(p) \sim_S I^+(p)$ for all $p\in V$, and these are
the unique S-relations (according to our definition (\ref{eSz}))
involving proper indecomposable sets. In the case $(P,F)\in
\hat{V}_{\emptyset}\times \check{V}_{\emptyset}\setminus
\{(\emptyset,\emptyset)\}$, we also include \be \label{eSz2} P\sim_S
\emptyset, \quad \quad (\hbox{resp.} \; \emptyset \sim_S F )\ee if
$P$ (resp. $F$) is a (non-empty, necessarily terminal)
indecomposable past (resp. future) set that  is not S-related by
(\ref{eSz}) to any other indecomposable set (notice that
$\emptyset$ is never S-related to itself).

Now, we are in conditions to introduce
the notion of c-completion at the point set level, according to \cite{FHSFinalDef}:
\begin{defi}\label{d1}
The {\em c-completion} $\overline{V}$ of a strongly causal spacetime $V$ is formed by all
the pairs $(P,F)\in
\hat{V}_{\emptyset}\times\check{V}_{\emptyset}$ with
$P\sim_{S} F$. The {\em c-boundary} $\partial V$ is defined as
$\partial V:=\overline{V}\setminus V$, under the identification $V\equiv
\{(I^{-}(p),I^{+}(p)): p\in V\}$.
\end{defi}

The chronological relation $\ll$ of the spacetime is extended to
the c-completion  in the following way: Two points $(P,F),
(P',F')\in \overline{V}$ are {\em chronologically related},
$(P,F)\overline{\ll} (P',F')$, if $F\cap P'\neq\emptyset$.
The situation is remarkably more complicated if one tries to define the extension $\overline{\leqslant}$ of the causal relation $\leqslant$. However, in the particular case of the spacetimes treated in this paper, the following criterium suffices (see the discussion in
\cite[Sect. 6.4]{FHSBuseman}, and references therein, for further details). Given two points $(P,F),
(P',F')\in \overline{V}$ with, either $P\neq\emptyset$ or $F'\neq
\emptyset$:
\[
P\subset P'\;\;\hbox{and}\;\; F'\subset F \Rightarrow
(P,F)\overline{\leqslant} (P',F').
\]
Moreover, we will
say that two different pairs in $\overline{V}$ are {\em
horismotically related} if they are causally but not
chronologically related.

\smallskip

Finally, the topology of the spacetime is also extended to the
c-completion by means of the so-called {\em chronological topology} ({\em chr. topology}, for short). This is a {\em sequential} topology defined
in terms of the following
{\em limit operator} $L$ for $\overline{V}$ (see \cite[Section 2]{FHSHaus} for an introduction to sequential topologies): given a sequence
$\sigma=\{(P_{n},F_{n})\}_n\subset\overline{V}$,

\begin{equation}
  \label{eq:29}
(P,F)\in L(\sigma)\iff\left\{ \begin{array}{ccc} P\in \hat{L}(\{P_n\}_n) & \hbox{whenever} &  P\neq \emptyset\\ F\in \check{L}(\{F_n\}_n) & \hbox{whenever} & F\neq \emptyset, \end{array}\right.
\end{equation}
 where
\begin{equation}\label{limcrono}
\begin{array}{c}
\hat{L}(\{P_{n}\}_n):=\{P'\in\hat{V}: P'\subset {\mathrm
LI}(\{P_{n}\}_n)\;\;\hbox{and}\;\; P'\;\;\hbox{is a maximal IP into}\;\; {\mathrm LS}(\{P_{n}\}_n)\} \\
\check{L}(\{F_{n}\}_n):=\{F'\in\check{V}: F'\subset {\mathrm
LI}(\{F_{n}\}_n)\;\;\hbox{and}\;\; F'\;\;\hbox{is a maximal IF
into}\;\; {\mathrm LS}(\{F_{n}\}_n)\}
\end{array}
\end{equation}
(LI and LS are the usual point set inferior and superior limits of
sets).
Concretely, the {\em closed sets} for the chr. topology are those subsets $C\subset V$ satisfying that $L(\sigma)\subset C$ for any sequence $\sigma\subset C$. Note that a topology on the future (resp. past) c-completion
$\hat{V}$ (resp. $\check{V}$) can be defined in a similar way,
just by using the limit operator $\hat{L}$ (resp. $\check{L}$)
instead of $L$. In this case, the resulting topology, which also
extends the topology of the spacetime, is called the {\em future}
(resp. {\em past}) {\em chronological topology}.
\begin{rem}\label{propsimplepunt} {\rm We emphasize
the following natural properties about the chronological topology:
\begin{itemize}
\item[(1)] The chronological topology (as well as the future and
past ones) is sequential and $T_1$ (see \cite[Prop. 3.39 and 3.21]{FHSFinalDef}), but may
be non-Hausdorff.

\item[(2)] Clearly, if $(P,F)\in L(\{(P_{n},F_{n})\}_n)$ then
$\{(P_{n},F_{n})\}_n$ converges to $(P,F)$. When the  converse
happens, $L$ is called {\em of first order} (see \cite[Section
3.6]{FHSFinalDef}).

\item[(3)] Given a pair $(P,F)\in \partial V$, any timelike curve
defining $P$ (or $F$) converges to $(P,F)$ with the chronological
topology (see \cite[Th. 3.27]{FHSFinalDef}).
\end{itemize}
}
\end{rem}

\smallskip

There are several subtleties involving the definition of the c-boundary which are essentially associated to the following facts: on one hand, a TIP (or TIF) may not determine a unique pair in the c-boundary; on the other hand, the topology does not always agree with the $S$-relation, in the sense that, for $S$-related elements as above:

$$P\in \hat{L}(\{P_n\}_n)\not\Leftrightarrow F\in \check{L}(\{F_n\}_n).$$ This
makes natural to consider the following special cases:
\begin{defi}\label{simpletop}
A spacetime $V$ has a c-completion $\overline{V}$ which is {\em
simple as a point set} if each TIP (resp. each TIF) determines a
unique pair in $\partial V$.
Moreover, the c-completion
is {\em simple} if it is simple as a point set and also {\em
topologically simple}, i.e. $(P,F)\in L(\{(P_{n},F_{n})\}_n)$ holds when
either $P\in \hat{L}(\{P_{n}\}_n)$ or $F\in \check{L}(\{F_{n}\}_n)$.
\end{defi}


\subsection{Case of interest: Generalized Robertson-Walker model}\label{sec:Robertson}

Let us restrict our attention to the future c-completion of Robertson-Walker models. In order to obtain it, we will reproduce the study developed in \cite[Section 3]{FHSIso2} adapted to this particular setting.

\smallskip

Let $(V,g)$ be a {\em Generalized Robertson-Walker model}, that is, $V=(a,b)\times M$
and
\[
g=-dt^2+\alpha g_M,
\]
where $\alpha:(a,b)\rightarrow (0,\infty)$ is a positive smooth function and $(M,g_{M})$ is a Riemannian manifold. This spacetime will be denoted by $(a,b) \times_{\alpha} M$ for short. Assume that the warping function $\alpha$ satisfies the following integral condition:
\begin{equation}
  \label{eq:47}
  \int_{\C}^{b}\frac{1}{\sqrt{\alpha(s)}}ds=\infty,\quad a<\C<b.
\end{equation}

\begin{rem}
  The only difference between the spacetime model studied in \cite{FHSIso2} and the one considered here is that  the temporal component $\R$ and the integral conditions

  \[
   \int^{\infty}_{0}\frac{1}{\sqrt{\alpha(s)}}ds=\int^{0}_{-\infty}\frac{1}{\sqrt{\alpha(s)}}ds=\infty
    \]
  has been replaced by a general interval $(a,b)$ and just the integral condition \eqref{eq:47}.
Nevertheless, the results established in this section are easily deducible by simple adaptations of the corresponding proofs in \cite{FHSIso2}. We leave the details to the reader interested on the subject.
\end{rem}

The chronological relation can be characterized in terms of the warping function $\alpha$ and the distance $d$ associated to $(M,g_M)$ as follows:
\begin{equation}
  \label{eq:26}
  \begin{array}{rl}
    (t^o,x^o)\ll (t^e,x^e) \iff   & d(x^o,x^e)<\displaystyle\int_{t^o}^{t^e}\frac{1}{\sqrt{\alpha(s)}}ds \\ & \\
    \iff & \displaystyle \int_{\C}^{t^o}\frac{1}{\sqrt{\alpha(s)}}ds<\int_{\C}^{t^e}\frac{1}{\sqrt{\alpha(s)}}ds-d(x^o,x^e).
  \end{array}
\end{equation}
Take a future-directed timelike curve $\gamma:[\omega,\Omega)\rightarrow V$, which can be expressed without loss of generality as $\gamma(t)=(t,c(t))$. The function
\begin{equation}
  \label{eq:25}
t\mapsto \int_{\C}^{t}\frac{1}{\sqrt{\alpha(s)}}ds-d(\cdot,c(t))
\end{equation}
is increasing with $t$ (see \cite[Prop. 3.1]{FHSIso2} for details). In particular, from \eqref{eq:26},
  \begin{equation}
    \label{eq:27}
  \begin{array}{rl}
      P=& I^-(\gamma)\\ = & \left\{(t^o,x^o)\in V:\displaystyle\int_{\C}^{t^o}\frac{1}{\alpha(s)}ds<\lim_{t^e\rightarrow \Omega}\left(\int_{\C}^{t}\frac{1}{\sqrt{\alpha(s)}}ds-d(x^o,c(t)) \right)\right\}.
    \end{array}
  \end{equation}
    Therefore, $P=P(b_{c})$, where $b_{c}(\cdot):=\lim_{t\rightarrow \Omega}\left(\displaystyle\int_{\C}^{t}\frac{1}{\sqrt{\alpha(s)}}ds-d(\cdot,c(t)) \right)$ is the {\em Busemann function associated to the curve $c$} and
    \begin{equation}
      \label{eq:28}
     P(f):=\{ (t^o,x^o)\in V:\int_{\C}^{t^o}\frac{1}{\sqrt{\alpha(s)}}ds<f(x^o)\}.
   \end{equation}

      Summarizing, the future c-completion $\hat{V}$, i.e. the set of all IPs, can be identified with the set of all Busemann functions on $M$. So, if we denote by $B(M)$ the set of all finite Busemann functions, it follows that
      \[
\hat{V}\equiv B(M)\cup \{\infty\},
        \]
where $\infty$ represents the constantly infinite Busemann function, which is associated to the TIP $P(\infty)=V=i^+$.

        Next, let us write $\hat{V}=\left(\hat{V}\setminus \hat{\partial}^{\B}V\right) \cup \hat{\partial}^{\ncambios{\B}}V$, where $\hat{\partial}^{\ncambios{\B}}V$ denotes the TIPs obtained from inextensible future-directed timelike curves with divergent timelike component ($\Omega=b$). The finite Busemann functions associated to these curves are called {\em proper}, and the set of all of them is denoted by ${\cal B}(M)$. So,
        \[
\hat{\partial}^{\ncambios{\B}}V\equiv {\cal B}(M)\cup \{\infty\}.
          \]


        In order to rewrite this set in a more appealing way, consider the quotient space
        \[
 \partial_{\cal B} M:={\cal B}(M)/\R
            \]
            where  two Busemann functions are $\R$-related if they differ only by a constant.
            Then, we can write
            \[
\hat{\partial}^{\ncambios{\B}} V\equiv \left(\R\times\partial_{\cal B} M \right)\cup \{\infty\},
              \]
              and so, we can see the future c-boundary as a cone with base $\partial_{\cal B} M$ and apex $\{\infty\}$. This picture is reinforced by the fact that the generatrix lines of the cone are shown to be horismotic, that is, each couple of points on the same generatrix line are horismotically related (see \cite[Section 3]{FHSIso2}).


       The remaining set $\left(\hat{V}\setminus \hat{\partial}^{\ncambios{\B}}V\right)$ is formed by IPs obtained as the past of future-directed timelike curves $\gamma:[\omega,\Omega)\rightarrow V$, $\gamma(t)=(t,c(t))$, with $\Omega<b$. It can be proved that, in this case, $c(t)\rightarrow x^*\in M^C$, where $M^C$ denotes the Cauchy completion of $(M,g_{M})$, and so,

         \begin{equation}
           \label{eq:46}
b_{c}(\cdot)=d_{(\Omega,x^*)}(\cdot):=\int_{\C}^{\Omega}\frac{1}{\sqrt{\alpha(s)}}ds-d(\cdot,x^*)
         \end{equation}
(see \cite[(3.7)]{FHSIso2}). In conclusion, we have the following identification
       \begin{equation}
            \label{eq:19}
       \hat{V}\setminus \hat{\partial}^{\ncambios{\B}}V\equiv (a,b) \times M^C,
          \end{equation}
          which implies,
          \[
\hat{V}=
\left(\hat{V}\setminus \hat{\partial}^{\ncambios{\B}}V\right)\cup \hat{\partial}^{\ncambios{\B}}V
\equiv
\left((a,b)\times M^C\right) \cup \left(\R\times \partial_{\cal B} M \right) \cup \{\infty\}.
            \]
            and
            \[
            \hat{\partial}V=((a,b) \times\partial^C M)\cup (\R\times \partial_{{\cal B}}M)\cup \{\infty\}.
            \]

Note that, given $P= P(b_{c})$ and $P_n=P(b_{c_n})$,
              \begin{equation}
                \label{eq:50}
                \begin{array}{c}
                P\subset {\rm LI}(\{P_n\}_n) \iff b_{c}\leq \liminf_n(\{b_{c_n}\}_n)\\
                \left(\hbox{resp. }P\subset {\rm LS}(\{P_n\}_n) \iff b_{c}\leq \limsup_n(\{b_{c_n}\}_n)\right).
               \end{array}
              \end{equation}
              This property joined to the identification between $\hat{V}$ and $B(M)$ described above, suggests to translate
            the future chronological topology on $\hat{V}$ into a (sequential) topology on $B(M)$, which is also called {\em future chronological topology}
            (see \cite[Section 3.3]{FHSIso2}. The limit operatior for this topology, also denoted by $\hat{L}$, is defined as follows: 
            \begin{equation}
              \label{eq:22}
              f\in \hat{L}(\{f_n\}_n) \iff \left\{
                \begin{array}{l}
                  (a)\;\; f\leq {\liminf}_n f_n \hbox{ and}\\
                  (b)\;\; \forall g\in B(M) \hbox{ with $f\leq g\leq \limsup_n f_n,$ it is $g=f$.}
                \end{array}\right.
            \end{equation}

            The following result establishes the relation between this topology and the pointwise topology on $B(M)$ (see \cite[Prop. 3.2]{FHSIso2} and \cite[Prop. 5.29]{FHSBuseman}):
            \begin{prop}\label{prel:PropToponefibre}
              Consider $\{f_n\}_n\subset B(M)$ a sequence which converges pointwise to a function $f\in B(M)$. Then, $f$ is the unique future chronological limit of $\{f_n\}_n$. In particular, if $\{f_n\}_n=\{d_{(\Omega^n,x^n)}\}_n$ with $\Omega^n\rightarrow \Omega$ and $x^n\rightarrow x (\in M^C)$, then $f=d_{(\Omega,x)}\in \hat{L}(\{f_n\}_n)$ is the unique future chronological limit of $\{f_n\}_n$.

              Moreover, if $M^C$ is locally compact, then the following converse follows: if $f=d_{(\Omega,x)}\in \hat{L}(\{f_n\}_n)$, then for $n$ big enough $\{f_n\}_n=\{d_{(\Omega^n,x^n)}\}_n$ for some $\Omega^n\in \R$ and $x^n\in M^C$ satisfying that $\Omega^n\rightarrow \Omega$ and $x^n\rightarrow x (\in M^C)$.
            \end{prop}

            \smallskip

            Finally, note that the study of the past c-completion is very similar, just with some minor changes. First, the following integral condition is imposed:
              \[
\int_{a}^\C \frac{1}{\sqrt{\alpha(s)}}ds=\infty.
                \]
              Given a past-directed timelike curve $\gamma:[\omega,-\Omega)\rightarrow V$, $\gamma(t)=(-t,c(t))$, then $I^+(\gamma)=F(-b^-_{c})$ where
            \[
F(f):=\{ (t^o,x^o)\in V:\int_{\C}^{t^o}\frac{1}{\sqrt{\alpha(s)}}ds>f(x^o)\}.
              \]
              Moreover, the backward Busemann functions are written now as
              \[
              b^-_{c}(\cdot):=\lim_{t\rightarrow -\Omega}\left(\int^{\C}_{-t}\frac{1}{\sqrt{\alpha(s)}}ds -d(\cdot,c(t))\right). 
              \]
              The space of finite backward Busemann functions coincide with $B(M)$, so there is a natural bijection between the future and past c-completions. Moreover,  when $a<-\Omega$, then $c(t)\rightarrow x^*\in M^C$ and the backward Busemann function becomes (compare with \eqref{eq:46})
              \begin{equation}
                \label{eq:48}
                b^-_{c}(\cdot)=d^-_{(\Omega,x^*)}(\cdot)=\int^{\C}_{\Omega}\frac{1}{\sqrt{\alpha(s)}}ds-d(\cdot,x^*).
              \end{equation}
In conclusion, one deduces
\[
\check{V}\setminus \check{\partial}^{a}V\equiv (a,b) \times M^C,
  \]
  and then,
  \[\begin{array}{rl}
      \check{V} \equiv & B(M)\cup \{-\infty\}   \\
         \equiv  & \left( (a,b) \times M^C  \right) \cup \left({\cal B}(M)\cup \{-\infty\} \right) \\ \equiv & \left( (a,b) \times M^C  \right) \cup \left(\R\times \partial_{\cal B}(M)\right) \cup \{-\infty\}.
  \end{array}
  \]


%% file: Preliminaries.tex
\section{The causal structure of doubly warped spacetimes}
\label{sec:chronologicalrelation}
%

In this section we are going to characterize the chronological and causal relations in doubly warped spacetimes. First, recall that a {\em \multiwarped spacetime} is a multiwarped spacetime $(V,g)$ as in (\ref{eqqq}) with two fibers ($n=2$), that is,

\begin{equation}
  \label{eq:1-aux}
  V:= (a,b)\times M_{1} \times M_{2}\quad\hbox{and}\quad
g=-dt^{2}+\alpha_{1}g_{1}+\alpha_{2}g_{2}.
\end{equation}

Take $(t^e,x^e)\in V$ and $x^o\in M:=M_1\times M_2$. Denote by $C(x^o,x^e)$ the set of smooth curves in $M$ connecting $x^o$ with $x^e$. Given $c=(c_1,c_2)\in C(x^o,x^e)$, consider the unique future-directed lightlike curve $\rho:[s^o,s^e]\rightarrow V$ with $\rho(s)=(\tau_{c,t^e}(s),c(s))$ and $\tau_{c,t^e}(s^e)=t^e$. From the metric expression in (\ref{eq:1-aux}), the component $\tau_{c,t^e}(s)$ is determined by the Cauchy problem
\[
-\dot{\tau}_{c,t^e}^2+\alpha_1(\tau_{c,t^e})g_{1}(\dot{c}_1,\dot{c}_1)+\alpha_2(\tau_{c,t^e})g_{2}(\dot{c}_2,\dot{c}_2)=0,\qquad
    \tau_{c,t^e}(s^e)=t^e.
  \]
Consider the functional
\[\J_{x^o,(t^e,x^e)}: C(x^o,x^e) \rightarrow (a,b), \quad c \mapsto \tau_{c,t^e}(s^o).\]
A direct computation shows that $(t^o,x^o)\ll (t^e,x^e)$ if, and only if, there exists $c\in C(x^o,x^e)$ such that $t_o<\J_{x^o,(t^e,x^e)}(c)$. This property suggests the following definition for the {\em departure time function}:
\[
T:M\times \left((a,b)\times M\right)\rightarrow (a,b),\qquad T(x^{o},(t^{e},x^{e})):= {Sup}_{C}\J_{x^o,(t^e,x^e)}
\]
(compare with \cite[Section 2.9]{Perlick2004} and \cite[Section 4]{FS2}). By construction, this function characterizes the chronological relation in $(V,g)$, as follows:
\begin{equation}\label{e0}
(t^{o},x^{o}) \ll (t^{e},x^{e}) \;\; \Longleftrightarrow \;\;
t^{o}<T(x^{o},(t^{e},x^{e})).
\end{equation}
In particular, the chronological past of a given point $(t^e,x^e)$ is given by
\[
I^-\left((t^e,x^e) \right):=\{(t,x)\in (a,b)\times M: t<T(x,(t^e,x^e)) \}.
  \]
Given a future-directed timelike curve $\gamma(t)=(t,c(t))$, $t\in [\omega,\Omega)$, and a point $x\in M$, the transitivity of the chronological relaction $\ll$ ensures that the function $T(x,\gamma(t))$ is increasing on $t$. Hence, the chronological past of $\gamma$ can be written as
  \[
I^-(\gamma)=\{(s,x)\in (a,b)\times M: s<b_c(x):=lim_{t\rightarrow b}T(x,\gamma(t))\}.
    \]

Next, let us characterize the departure time function, and so, the chronological relations (recall (\ref{e0})), in terms of some integral conditions involving the warping functions $\alpha_i$ and the Riemannian distances $d_i$ associated to the fibers $(M_i,g_i)$, $i=1,2$.
To this aim, let us consider a future-directed causal curve $\gamma:  I \rightarrow V$,
$\gamma(s)=(t(s),c_{1}(s),c_{2}(s))$. From the metric expression in \eqref{eq:1-aux}:
\[
\frac{dt}{ds}(s)=\sqrt{-D+\frac{\mu_{1}}{\alpha_{1}\circ
t}+\frac{\mu_{2}}{\alpha_{2}\circ t}}(s),
\]
where $D:=g(d\gamma/ds,d\gamma/ds)\leq 0$ and
$\mu_{i}:=(\alpha_{i}\circ t)^2 g_{i}(dc_{i}/ds,dc_{i}/ds)$,
$i=1,2$. From the Inverse Function Theorem, previous formula translates into
\[
\frac{ds}{dt}(t)=\left(-(D\circ s)+\frac{\mu_{1}\circ s}{\alpha
_{1}}+\frac{\mu_{2}\circ s}{\alpha _{2}}\right)^{-1/2}.
\]
Therefore, if we denote $t^{o}=t(s^{o})$, $t^{e}=t(s^{e})$, we
deduce
\begin{equation}\label{eq:3}
\begin{array}{c}
\hbox{length}\left(c_{i}\mid_{[s^{o},s^{e}]}\right)=\int_{s^{o}}^{s^{e}}\sqrt{g_{i}(\dot
c_{i}, \dot c_{i})} ds=\int_{t^o}^{t^e}\sqrt{g_{i}(\dot c_{i}, \dot
c_{i})} \frac{ds}{dt} dt \qquad\qquad\qquad\quad \\
\;\quad\qquad\qquad\qquad =\int_{t^o}^{t^e}\frac{\sqrt{\mu_{i}\circ
s}}{\alpha_i(t)}\left(-(D\circ s)+\frac{\mu_{1}\circ
s}{\alpha_{1}(t)}+\frac{\mu_{2}\circ s}{\alpha_{2}(t)}\right)^{-1/2}dt
\qquad\hbox{for}\;\; i=1,2.
\end{array}
\end{equation}
In the particular case of being $\gamma$ a lightlike geodesic we have: (i) $D=0$ (lightlike character of $\gamma)$, (ii) $\mu_i\circ s$ are constants and (iii) $c_i$ are (pre-)geodesics on the corresponding Riemannian manifold $(M_i,g_i)$ (geodesic character of $\gamma$). So, from (\ref{eq:3}), one deduces (see \cite[Theorem 2]{FS} for details):
\begin{prop}\label{thm:characluzgeodesics}
  Let $(V,g)$ be a {\multiwarped} spacetime as in (\ref{eq:1-aux}) with (weakly) convex fibers (i.e., satisfying that any pair of points can be joined by some minimizing geodesic). Consider two distinct points $(t^o,x_1^o,x_2^o),(t^e,x_1^e,x_2^e)\in V$ with $t^o<t^e$. Then, the following statements are equivalent:
  \begin{itemize}
  \item[(a)] There exists a lightlike geodesic joining $(t^o,x_1^o,x_2^o)$ and $(t^e,x_1^e,x_2^e)$.
  \item[(b)] There exist $\mu_1,\mu_2 \ncambios{\geq}0$ with $\mu_1+\mu_2=1$ such that
    \[
\Integral{t^o}{t^{e}}{\mu_{i}}{i}{\mu_{k}}=d_{i}(x^{o}_{i},x^{e}_{i})\qquad\hbox{for}\;\;
i=1,2;
      \]

  \end{itemize}

\end{prop}

We are now in conditions to establish the characterization of the chronological relation.
\begin{prop}\label{c0}
Let $(V,g)$ be a \multiwarped spacetime as in (\ref{eq:1-aux}), and $(t^{o},x^{o}), (t^{e},x^{e})\in V$ with $x^{o}\neq
x^{e}$. The following conditions are equivalent:
\begin{itemize}

\item[(i)]  $(t^{o},x^{o})\ll (t^{e},x^{e})$; or, equivalently, $t^o<T(x^o,(t^e,x^e))$ (recall (\ref{e0}));
\item[(ii)] $T(x^o,(t^e,x^e))$ is the unique real value $T\in (a,b)$
with $t^{o}<T<t^{e}$ such that, for some (unique) positive constants $\mu_{1},\mu_2 \ncambios{\geq}
0$, with $\mu_{1}+\mu_{2}=1$, it satisfies
\begin{equation}\label{ee2}
\Integral{T}{t^{e}}{\mu_{i}}{i}{\mu_{k}}=d_{i}(x^{o}_{i},x^{e}_{i})\qquad\hbox{for}\;\;
i=1,2;
\end{equation}

\item[(iii)] there exist strictly positive constants $\mu'_{1},\mu'_{2}> 0$, with $\mu'_1+\mu'_2=1$,
such that
\begin{equation}\label{ee2''}
\Integral{t^{o}}{t^{e}}{\mu_{i}'}{i}{\mu_{k}'}>
d_{i}(x^{o}_{i},x^{e}_{i})\qquad\hbox{for $i=1,2$}.
\end{equation}
\end{itemize}
\end{prop}
\begin{proof}
The implication $(ii)\Rightarrow (iii)$ is trivial unless some $\mu_i$ is equal to $0$. So, assume for instance that $\mu_1=0$ (and so, $\mu_2=1$). Then, \eqref{ee2} becomes
\[
\left\{
  \begin{array}{l}
    0=d_1(x_1^o,x_1^e)\\
    \\
    \displaystyle \int_{T}^{t^e}\frac{1}{\sqrt{\alpha_2(s)}}ds=d_2(x_2^o,x_2^e).
  \end{array}
\right.
  \]
  By continuity, we can modify slightly $\mu_1$, $\mu_2$, to obtain strictly positive $\mu'_1,\mu'_2$, with $\mu'_1+\mu'_2=1$, such that
 \[
    \left\{
      \begin{array}{l}\displaystyle\Integral{t^{o}}{t^{e}}{\mu_{1}'}{1}{\mu_{k}'}>0= d_1(x_1^o,x_1^e)\\
      \\
      \displaystyle\Integral{t^{o}}{t^{e}}{\mu_{2}'}{2}{\mu_{k}'}> d_2(x_2^o,x_2^e),
      \end{array}\right.
    \]
as desired.

For the implication $(iii) \Rightarrow (i)$, denote
  \[
L^\epsilon_i:=\Integral{t^o+\epsilon}{t^e}{\mu_{i}'}{i}{\mu_{k}'},\quad\hbox{for $\epsilon>0$.}
    \]
 Take $\epsilon>0$ small enough so that $t^o+\epsilon<t^e$ and the inequalities in \eqref{ee2''} still hold for $t^o+\epsilon$ instead of $t^o$. Since $L_i^{\epsilon}>d_i(x_i^o,x_i^e)$, there exist curves $y_i:[s^o,s^e]\rightarrow M_i$, with $y_{i}(s^{o})=x_{i}^{o}$ and $y_{i}(s^{e})=x^{e}_{i}$, such that $length(y_{i})=L^\epsilon_{i}$, $i=1,2$. Consider the lightlike curve
$\rho(s)=(\tau(s),\overline{y}_{1}(s),\overline{y}_{2}(s))$,
with $\overline{y}_{i}$ reparametrizations of $y_{i}$,
constructed by
requiring
\[
\left\{\begin{array}{l}\dot{\tau}=\sqrt{\sum_{i=1}^{2}\frac{\mu_{i}'}{\alpha_{i}\circ\tau}}
\\ \tau(s^{e})=t^{e}
\end{array}\right.,\qquad
\left\{\begin{array}{l}g_{i}(\dot{\overline{y}}_{i},\dot{\overline{y}}_{i})=\frac{\mu_{i}'}{(\alpha_{i}\circ
\tau)^{2}} \\
\overline{y}_{i}(s^{e})=x^{e}_{i}\end{array}\right.
\qquad\hbox{for}\;\; i=1,2.
\]
Then, by applying \eqref{eq:3} to the lightlike curve $\rho$ (in particular, $D=0$), we deduce
\[
\hbox{length}(\overline{y}_{i}\mid_{[\tau^{-1}(t^o+\epsilon),s^{e}]})=\Integral{t^o+\epsilon}{t^{e}}{\mu_{i}'}{i}{\mu_{k}'}
=L^\epsilon_{i}=\hbox{length}(y_{i}).
\]
Therefore, $\rho(s)$ is a lightlike curve joining $(t^o+\epsilon,x^{o})$ with
$(t^{e},x^{e})$, and so, these points are causally related. Since $(t^o,x^o)\ll (t^o+\epsilon,x^o)$, necessarily $(t^o,x^o)\ll (t^{e},x^{e})$.

Finally, for the implication $(i) \Rightarrow (ii)$, let us show first that if $T$
satisfies (\ref{ee2}) then $T\leq T(x^o,(t^e,x^e))$. So, assume that (\ref{ee2}) holds. 
Take any sequence $\epsilon_{m} \searrow 0$ and define
\[
L_{i,m}:=\Integral{T-\epsilon_m}{t^{e}}{\mu_{i}}{i}{\mu_{k}}\quad\hbox{i=1,2.}
\]
We have that $L_{i,m}>d_{i}(x_{i}^o,x_{i}^{e})$ for all $i$. The implication (iii)$\Rightarrow$(i), which has been proved before, ensures that
$(T-\epsilon_{m},x^o) \ll (t^e,x^e)$ for all $m$. Therefore, from the definition of $T(x^o,(t^e,x^e))$, $T-\epsilon_{m}<T(x^o,(t^e,x^e))$ for all $m$, and then, $T \leq T(x^o,(t^e,x^e))$.


Next, it is sufficient to prove that some value $T$
verifying (\ref{ee2}) always exists, and necessarily $T\geq T(x^o,(t^e,x^e))$. Let $t'<T(x^o,(t^e,x^e))$, and thus, $(t',x^{o})\ll (t^{e},x^{e})$.
Let $\gamma:[t',t^e]\rightarrow V$,
$\gamma(t)=(t,c_{1}(t), c_{2}(t))$, be a timelike curve such
that $\gamma(t')=(t',x^{o})$ and $\gamma(t^{e})=(t^{e},x^{e})$.
Consider real curves $\overline{c}_{i}$, $i=1,2$, such that
\[
\left\{\begin{array}{l} 0\leq\dot{\overline{c}}_{i}(t)\leq
\sqrt{g_{i}(\dot{c}_{i}(t),\dot{c}_{i}(t))} \\
\overline{c}_{i}(t')=0 \\
\overline{c}_{i}(t^{e})=d_{i}(x_{i}^{o},x_{i}^{e})
\end{array}\right. \qquad\hbox{for}\;\; i=1,2.
\]
Then, $\overline{\gamma}(t)=(t,\overline{c}_{1}(t),\overline{c}_{2}(t))$
becomes a future directed timelike curve in the globally hyperbolic \multiwarped
spacetime with convex fibers
$V'=(\ncambios{(a,b)} \times \R^{2},-dt^{2}+\alpha_{1}dx_{1}^{2}+\alpha_{2}dx_{2}^{2})$ joining $\overline{\gamma}(t')=\point{t'}{0}{0}$ with $\overline{\gamma}(t^e)=\point{t^e}{d_{1}(x_{1}^o,x_{1}^{e})}{d_{2}(x_{2}^o,x_{2}^e)}$, i.e., \[\overline{\gamma}(t')=(t',0,0)\ll (t^e,d_{1}(x_{1}^o,x_{1}^{e}),d_{2}(x_{2}^o,x_{2}^{e}))=\overline{\gamma}(t^e).\]
Consider $T>t'$ such that $(T,0,0)\leq \overline{\gamma}(t^e)$ but $(T,0,0)\not\ll \overline{\gamma}(t^e)$. From
Avez and Seifert's result, there exists some lightlike geodesic in $V'$
joining both points. Now, from Prop. \ref{thm:characluzgeodesics} applied to this lightlike geodesic, there exist unique positive constants $\ncambios{\mu_1,\mu_2 \geq}0$, with $\mu_1+\mu_2=1$, such that
\[
\Integral{T}{t^{e}}{\mu_{i}}{i}{\mu_{k}}=|d_{i}(x^{o}_{i},x^{e}_{i})-0|=d_{i}(x^{o}_{i},x^{e}_{i}) \qquad\hbox{for}\;\;
i=1,2.
  \]
Finally, since $t'<T$ for all $t'<T(x^o,(t^e,x^e))$, necessarily $T(x^o,(t^e,x^e))\leq T$, which concludes the proof.

\end{proof}


Let us consider now the characterization of the causal relation (see \cite[Theorem 2(2)]{FS}).
 \begin{defi}
A Riemannian manifold $(N,h)$ is $L$-{\em convex} if any pair of points $p,q\in N$ with $d_h(p,q)<L$ can be joined by a minimizing geodesic.
\end{defi}
\begin{prop}
\label{p2'}
Let $(V,g)$ be a {\multiwarped} spacetime as in (\ref{eq:1-aux}) whose fibers $(M_i,g_i)$ are $L_i$-convex for $i=1,2$. Consider two points $\point{t^{o}}{x^o_{1}}{x^o_2}, \point{t^{e}}{x_1^{e}}{x_2^e} \in V$, with $t^o \leq t^e$, satisfying $d(x_i^o,x_i^e)<L_i$, $i=1,2$. Then, the following conditions are equivalent:
\begin{itemize}
\item[(i)] the points are causally related,
$\point{t^{o}}{x_1^{o}}{x_2^o} \leq \point{t^{e}}{x_1^{e}}{x_2^e}$;

\item[(ii)] there exists a causal geodesic joining $\point{t^{o}}{x_1^{o}}{x_2^o}$ with $\point{t^{e}}{x_1^{e}}{x_2^e}$;

\item[(iii)] there exist constants $\mu'_{1},\mu'_{2}\geq 0$, $\mu'_1+\mu'_2=1$,
such that
\begin{equation}
\label{e2'''}
\Integral{t^{o}}{t^{e}}{\mu'_{i}}{i}{\mu'_{k}} \geq
d_{i}(x^{o}_{i},x^{e}_{i})\qquad\hbox{for}\;\;
i=1,2.
\end{equation}
\end{itemize}
Moreover, if the equalities hold in (\ref{e2'''}), then there is a lightlike and no timelike geodesic joining the points.
\end{prop}

%% file: hierarchy.tex
\section{Position into the causal ladder}
\label{sec:causalladder}
In order to have an idea of the goodness of the causality of doubly warped spacetimes, next we are going to
determine their position into the causal ladder. As we will see, this depends on the warping functions integrals and the convexity character of their Riemannian fibers.

Let us consider first a brief remainder of the main levels of the causal ladder. Each level corresponds with a causality condition which is strictly more restrictive than the previous one:
\begin{defi}\label{ant} A spacetime $(V,g)$ is
\begin{itemize}
\item {\em non-totally vicious} if $p\not\ll p$ for some $p\in V$.

\item {\em chronological} if it does not contain closed timelike curves.

\item {\em causal} if it does not contain closed causal curves.

\item {\em distinguishing} if whenever $I^+(p)=I^+(q)$ and $I^-(p)=I^-(q)$, necessarily $p=q$.

\item{\em strongly causal} if it does not contain ``nearly closed'' causal curves, i.e. for any open neighborhood $U$ of $p$ there exists some open neighborhood $V$ with $p\in V\subset U$ such that any timelike segment with extreme points in $V$ is contained in $U$.

\item {\em stably causal} if there exists some causal Lorentzian metric $g'$ on $V$ with $g<g'$, i.e., such that $g'(v,v)<0$ for any $v\in TV\setminus \{0\}$ with $g(v,v)\leq 0$. This is equivalent to the existence of some {\em global time function}, i.e., a function defined on the whole spacetime $(V,g)$ which is strictly increasing along each future-directed causal curve.

\item {\em causally continuous} if it is distinguishing and the set valued functions $I^{+}(\cdot)$ and $I^{-}(\cdot)$ are outer continuous (say, $I^{+}(\cdot)$ is {\em outer continuous at some} $p\in V$ if, for any compact subset $K\subset I^{+}(p)$ there exists an open neighborhood $U\ni p$ such that $K\subset I^{+}(q)$ for all $q\in U$). This is equivalent to being distinguishing and {\em reflecting}, 
i.e. for any pair of events $p,q \in V$, $I^{+}(q) \subset I^{+}(p)$ if and only if $I^{-}(p) \subset I^{-}(q)$.

\item {\em causally simple} if it is causal and $J^{\pm}(p)$ are closed sets for any $p\in V$.

\item {\em globally hyperbolic} if it is causal and $J^{+}(p)\cap J^{-}(q)$ are compact for any $p,q \in V$.
\end{itemize}
\end{defi}
It is direct from the very basic structure of \multiwarped spacetimes (\ref{eq:1-aux}) that $t:V \rightarrow (a,b)$ is a global time function (see \cite[Lemma 3.55]{beem}). Therefore, any \multiwarped spacetime is stably causal. The approach developed in previous section will allow to show that any \multiwarped spacetime is causally continuous as well. In fact:

\begin{thm}
Any \multiwarped spacetime $(V,g)$ as in (\ref{eq:1-aux}) is causally continuous.
\end{thm}

\begin{proof}  Since $(V,g)$ is stably causal, it is also distinguishing. So, it suffices to show that $(V,g)$ is reflecting. Let $\point{t^o}{x_{1}^o}{x_{2}^{o}}, \point{t^{e}}{x_{1}^{e}}{x_{2}^{e}} \in V$ be such that
$I^{+}(\point{t^{e}}{x_{1}^{e}}{x_{2}^{e}}) \subset I^{+}(\point{t^{o}}{x_{1}^{o}}{x_{2}^{o}})$, and let us prove that $I^{-}(\point{t^{o}}{x_{1}^{o}}{x_{2}^{o}})
\subset I^{-}(\point{t^{e}}{x_{1}^{e}}{x_{2}^{e}})$
(the converse is analogous). Consider the sequence $\{\point{t^{e}+1/n}{x_{1}^{e}}{x_{2}^{e}}\}_{n} \subset I^{+}(\point{t^{e}}{x_{1}^{e}}{x_{2}^{e}})$ 
and note that, by the hypothesis, this sequence also belongs to $I^{+}(\point{t^{o}}{x_{1}^{o}}{x_{2}^{o}})$.
Therefore, from Prop. \ref{c0}, there exist constants $\mu_{1}^{n},\mu_{2}^{n}>0$, with $\mu_{1}^{n} + \mu_{2}^{n}=1$, satisfying the following inequalities:
\begin{equation}
\label{eq2'}
\Integral{t^{o}}{t^{e}+1/n}{\mu_{i}^{n}}{i}{\mu_{k}^{n}}
> d_{i}(x_{i}^{o},x_{i}^{e})\qquad\hbox{for $i=1,2$}.
\end{equation}
Up to a subsequence, we can assume that $\{\mu_{i}^{n}\}_{n}$ converges to $\mu_i$, for all $i$, with $0 \leq \mu_{1},\mu_{2}\leq 1$ and $\mu_{1}+\mu_{2}=1$. Moreover,
\begin{equation}
\label{eq3'}
\left\{\sqrt{\mu_{i}^{n}}\alpha_{i}(s)^{-1}\left(\sum_{k=1}^2\mu_{k}^{n} \alpha_{k}(s)^{-1}\right)^{-1/2}\right\}_n \longrightarrow\sqrt{\mu_{i}}\alpha_{i}(s)^{-1}\left(\sum_{k=1}^2\mu_{k}\alpha_{k}(s)^{-1}\right)^{-1/2}
\end{equation}
uniformly on $[t^o,t^e+1]$. Therefore, from (\ref{eq2'}) and (\ref{eq3'}), we deduce
\[
\Integral{t^{o}}{t^{e}}{\mu_{i}}{i}{\mu_{k}} \geq d_{i}(x_{i}^{o},x_{i}^{e}),\qquad\hbox{for $i=1,2$.}
\]
If we consider $\point{t^{o}-1/n}{x_{1}^{o}}{x_{2}^o}$, and modify slightly $(\mu_{1},\mu_{2})$, by continuity we obtain new coefficients $(\mu_{1}',\mu_{2}')$, with $\mu'_{1},\mu'_{2}>0$ and $\mu'_{1}+\mu'_{2}=1$, such that
\[
\Integral{t^{o}-1/n}{t^e}{\mu'_{i}}{i}{\mu'_{k}}
>
d_{i}(x^{o}_{i},x^{e}_{i})\qquad\hbox{for $i=1,2$.}
\]
Again from Prop. \ref{c0}, we have $\point{t^o-1/n}{x_1^o}{x_2^o} \ll \point{t^e}{x_1^e}{x_2^e}$
for all $n$. So, taking into account that $I^-(\point{t^o}{x_1^o}{x_2^o})=\cup_{n \in {\mathbb N}}I^-(\point{t^o-1/n}{x_1^o}{x_2^o})$, we deduce the inclusion $I^-(\point{t^o}{x_1^o}{x_2^o} )\subset
I^-(\point{t^e}{x_1^e}{x_2^e})$, as required.
\end{proof}
\begin{thm}
\label{causi}
A \multiwarped spacetime $(V,g)$ as in (\ref{eq:1-aux}) is causally simple if and only if $(M_i,g_i)$ is $L_i$-convex for $L_i=\int_{a}^{b}\frac{1}{\sqrt{\alpha_{i}(s)}}ds$, $i=1,2$.
\end{thm}

\begin{proof} For the implication to the right, assume that $(V,g)$ is causally simple. We will prove that
$(M_1,g_1)$ is $L_1$-convex (the proof for the second fiber is analogous). Let $x^o_{1},x^e_{1} \in M_{1}$ with $0<d_1(x^o_{1},x^e_{1})<L_1$. Since
$\int_{a}^{b}\frac{1}{\sqrt{\alpha_1(s)}}ds=L_1 > d_1(x^o_{1},x^e_{1})$, there exists $a<\C_1<\C_2<b$ such that
\begin{equation}
\label{eq6'}
\int_{\C_1}^{\C_2} \frac{1}{\sqrt{\alpha_{1}(s)}}ds>d_{1}(x^o_{1},x^e_{1}).
\end{equation}
Fix $x_{2} \in M_{2}$ and consider
the points $\point{\C_1}{x_{1}^o}{x_{2}}$ and $\point{\C_2}{x_{1}^e}{x_{2}}$. Inequality (\ref{eq6'}) and Prop. \ref{c0} imply that
$\point{\C_2}{x_{1}^e}{x_{2}} \in I^{+}(\point{\C_1}{x_{1}^o}{x_{2}})$.
Since $\point{\C_1}{x_{1}^e}{x_{2}} \not \in I^+(\point{\C_1}{x_{1}^o}{x_{2}})$, there exists
$t^e \in \R$ such that $\point{t^e}{x^e_{1}}{x_{2}} \in \partial I^{+}(\point{\C_1}{x^o_{1}}{x_{2}})$,
i.e.,
\[
\begin{array}{c}
\point{t^e}{x^e_{1}}{x_{2}} \in \overline{I^{+}(\point{\C_1}{x^o_{1}}{x_{2}})} \setminus I^{+}(\point{\C_1}{x^o_{1}}{x_{2}}) \qquad\qquad\qquad \\ \qquad\qquad\qquad\qquad\qquad\qquad=J^{+}(\point{\C_1}{x^o_{1}}{x_{2}})\setminus I^{+}(\point{\C_1}{x^o_{1}}{x_{2}}),
\end{array}
\]
where, in the equality, we have used that $(V,g)$ is causally simple. Therefore, there exists a null geodesic
$\gamma(s)=(t(s),c_{1}(s),c_{2}(s))$ connecting $\point{\C_1}{x^o_{1}}{x_{2}}$ with $\point{t^e}{x^e_{1}}{x_{2}}$.
From Prop. \ref{p2'} there exist constants $\mu'_1, \mu'_2\geq 0$ such that the following inequalities hold:
\[
\begin{array}{c}
0<d_{1}(x^o_{1},x^e_{1})\leq \displaystyle\Integral{\C_1}{t^e}{\mu'_{1}}{1}{{\mu'}_{k}}=\hbox{length}_1(c_{1}),
\\ 0=d_{2}(x_{2},x_{2}) \leq \displaystyle\Integral{\C_1}{t^e}{\mu'_{2}}{2}{{\mu'}_{k}}=\hbox{length}_2(c_{2}).
%
%
\end{array}
\]
So, taking into account that $$\point{t^e}{x^e_{1}}{x_{2}} \not\in I^+(\point{\C_1}{x^o_{1}}{x_{2}}),$$ the second inequality in the first line must be an equality (recall Prop. \ref{c0}). In conclusion, $c_1(s)$ is a reparametrization of a minimizing geodesic of $(M_1,g_1)$, as required.

\smallskip

For the implication to the left, assume that $(M_i,g_i)$ is $L_i$-convex for $L_i=\int_{a}^{b}\frac{1}{\sqrt{\alpha_i(s)}}ds$, $i=1,2$. In order to prove that $(V,g)$ is causally simple,
take $\point{t^{e}}{x_{1}^{e}}{x_{2}^{e}} \in \overline{J^{+}(\point{t^{o}}{x_{1}^{o}}{x_{2}^{o}})}=\overline{I^{+}(\point{t^{o}}{x_{1}^{o}}{x_{2}^{o}})}$. Then, $I^+(\point{t^{e}}{x_{1}^{e}}{x_{2}^{e}}) \subset I^+(\point{t^{o}}{x_{1}^{o}}{x_{2}^{o}})$, and thus, $\point{t^{o}}{x_{1}^{o}}{x_{2}^{o}} \ll \point{t^{e}+1/n}{x_{1}^{e}}{x_{2}^{e}}$ for all $n$.
From Prop. \ref{c0}, there exist constants $0 <\mu_{1}^{n},\mu_{2}^{n} < 1$, with $\mu_{1}^{n}+\mu_{2}^{n}=1$ for all $n$, such that
\[
\Integral{t^o}{t^{e}+1/n}{\mu_{i}^{n}}{i}{\mu_{k}^{n}}
> d_{i}(x_{i}^{o},x_{i}^{e}),\quad i=1,2.
\]
Since $\{\mu_{i}^{n}\}_{n}$ converges (up to subsequence) to some $\mu_{i} \in [0,1]$ for $i=1,2$, with $\mu_{1}+\mu_{2}=1$, we have:
\[
\frac{\sqrt{\mu_{i}^{n}}}{\alpha_{i}(s)}\left(\sum_{k=1}^2 \frac{\mu_{k}^{n}}{\alpha_{k}(s)}\right)^{-1/2}\longrightarrow
\frac{\sqrt{\mu_{i}}}{\alpha_{i}(s)}\left(\sum_{k=1}^2 \frac{\mu_{k}}{\alpha_{k}(s)}\right)^{-1/2}\quad\hbox{uniformly on $[t^o,t^e+1]$.}
\]
Recalling now that all previous functions are bounded by the (Lebesgue) integrable function $g:[t^o,t^e+1]\rightarrow \R$,
    $g(t)=\alpha_i(t)^{-1/2}$, the Dominated Convergence Theorem ensures that:
\[
\Integral{t^o}{t^e}{\mu_{i}}{i}{\mu_{k}}
=\lim_{n\rightarrow\infty} \Integral{t^{o}}{t^{e}+1/n}{\mu_{i}^{n}}{i}{\mu_{k}^n} \geq d_{i}(x_{i}^{o},x_{i}^{e}).
\]
In particular,
\[
d_{i}(x_{i}^{o},x_{i}^{e})<\Integral{a}{b}{\mu_{i}}{i}{\mu_{k}} \le \int_{a}^{b}\frac{1}{\sqrt{\alpha_i(s)}}ds=L_i,\quad i=1,2.
\]
So, taking into account that $(M_i,g_i)$ are $L_i$-convex for $i=1,2$ we have that Prop. \ref{p2'} implies $\point{t^{e}}{x_{1}^{e}}{x_2^e} \in J^{+}(\point{t^{o}}{x_{1}^{o}}{x_2^o})$, as required.
\end{proof}

The following example shows the tight character of Thm. \ref{causi}, in the sense that there may exist causally simple warped spacetimes with non-convex fiber (the extension to the case of two fibers is straightforward). In fact:
\begin{exe}

 \begin{figure}
\centering
\ifpdf
  \setlength{\unitlength}{1bp}%
  \begin{picture}(377.33, 110.89)(0,0)
  \put(0,0){\includegraphics{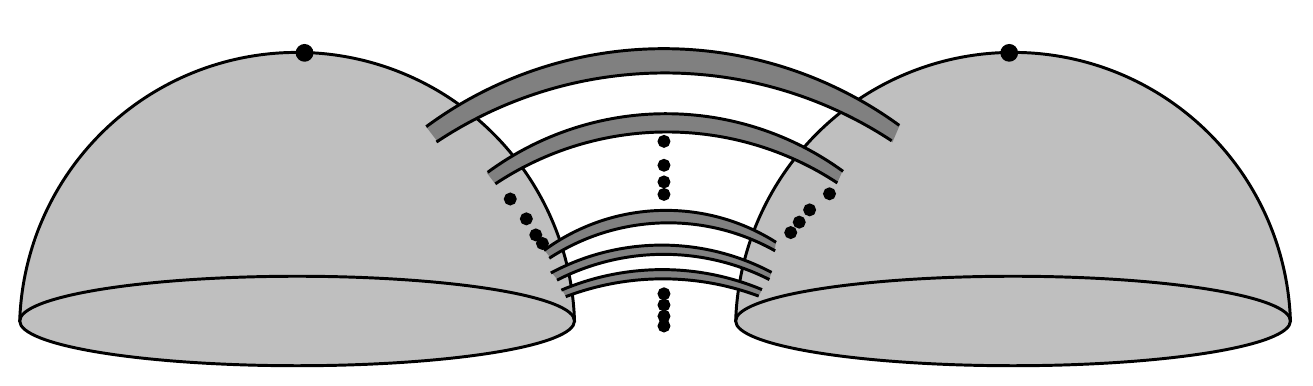}}
  \put(83.67,100.21){\fontsize{9.42}{11.71}\selectfont $x_0$}
  \put(286.41,100.21){\fontsize{9.42}{11.71}\selectfont $x_1$}
  \put(165.89,99.40){\fontsize{9.42}{11.71}\selectfont $T_1$}
  \put(166.29,78.44){\fontsize{9.42}{11.71}\selectfont $T_2$}
  \put(166.70,50.22){\fontsize{9.42}{11.71}\selectfont $T_n$}
  \put(47.67,69.16){\fontsize{13.76}{14.11}\selectfont $H_0$}
  \put(300.28,67.84){\fontsize{13.76}{14.11}\selectfont $H_1$}
  \end{picture}%
\else
  \setlength{\unitlength}{1bp}%
  \begin{picture}(377.33, 110.89)(0,0)
  \put(0,0){\includegraphics{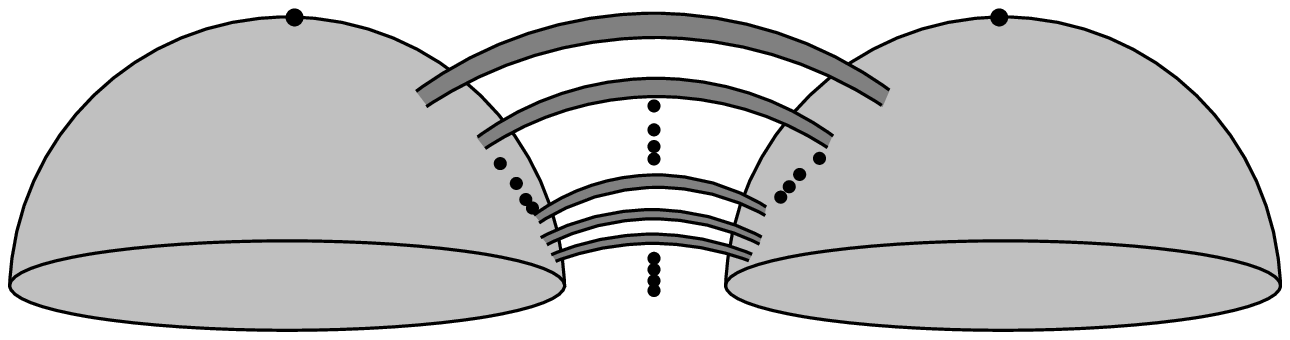}}
  \put(83.67,100.21){\fontsize{9.42}{11.71}\selectfont $x_0$}
  \put(286.41,100.21){\fontsize{9.42}{11.71}\selectfont $x_1$}
  \put(165.89,99.40){\fontsize{9.42}{11.71}\selectfont $T_1$}
  \put(166.29,78.44){\fontsize{9.42}{11.71}\selectfont $T_2$}
  \put(166.70,50.22){\fontsize{9.42}{11.71}\selectfont $T_n$}
  \put(47.67,69.16){\fontsize{13.76}{14.11}\selectfont $H_0$}
  \put(300.28,67.84){\fontsize{13.76}{14.11}\selectfont $H_1$}
  \end{picture}%
\fi

  \caption{\label{fig:1} Both hemispheres $H_0$ and $H_1$ are connected by a sequence of immersed tubes $\{T_n\}_n$, where a length-minimizing curve connecting the north pole $x_0$ of $H_0$ to the north pole $x_1$ of $H_1$ through $T_n$ has bigger length than a length-minimizing curve connecting the same points through $T_{n+1}$. This picture is based on \cite[Figure 1]{Bartolo2002}.}
\end{figure}

\footnote{We are thankful to Prof. Miguel Sánchez for bringing this example to our attention.}In \cite[Section 2.1]{Bartolo2002} the authors construct a Riemannian manifold $(M,g)$ containing two points $x_0, x_1\in M$ such that any geodesic $\gamma\subset M$ connecting them satisfies ${\mathrm length}(\gamma)>d(x_0,x_1)$. The example basically consists of two open hemispheres $H_0$, $H_1$ in $\R^3$ connected by a sequence of immersed tubes $(T_n)_n$ of decreasing lengths, and such that any curve joining the corresponding north poles $x_0$ and $x_1$ through $T_n$ is longer than a minimizing curve joining them through $T_{n+1}$ (see Figure \ref{fig:1}). It is assumed also that the lengths of these tubes converge to a number which is strictly positive. In particular, $x_{0}$ and $x_{1}$ cannot be joined by a minimizing geodesic, and thus,
$(M,g)$ is not convex. However, there exists some $\delta>0$ such that $(M,g)$ is $L$-convex for any $L\leq \delta$. Consider now the warped spacetime $V=\R\times_{\alpha}M$ with $\alpha:\R\rightarrow (0,\infty)$ satisfying $\int_{-\infty}^{+\infty}1/\sqrt{\alpha(s)}ds=L\leq\delta$. From Thm. \ref{causi}, $V$ is causally simple.
\end{exe}

Finally, for the sake of completeness, we include the following simple consequence of \cite[Th. 3.68]{beem}, whose implication to the left is reproved here by using the techniques developed in this paper:

\begin{thm}
A \multiwarped spacetime $(V,g)$ as in (\ref{eq:1-aux}) is globally hyperbolic if and only if $(M_{i},g_{i})$, $i=1,2$, are complete Riemannian manifolds.
\end{thm}

\begin{proof}
Assume that $(M_{i},g_{i})$, $i=1,2$, are complete. Since $(V,g)$ is causally continuous, and thus, causal, it suffices to prove that any causal diamond is sequentially compact (and thus, compact). Let $\{\point{t^n}{x_1^n}{x_2^n}\}_n$ be a sequence in $J^+(\point{t^o}{x_1^o}{x_2^o})\cap J^-(\point{t^e}{x_1^e}{x_2^e})$. Since the fibers are complete, they are convex, and so, we can apply Prop. \ref{p2'}. Hence, there exist constants $0\leq\mu_{1}^{n},\mu_{2}^{n}\leq 1$, $0\leq\overline{\mu}_{1}^{n},\overline{\mu}_{2}^{n}\leq 1$ with $\mu_{1}^{n}+\mu_{2}^{n}=1=\overline{\mu}_{1}^{n}+\overline{\mu}_{2}^{n}$ for all $n$, such that
\[
  \begin{array}{l}
\displaystyle\Integral{t^{o}}{t^n}{\mu_{i}^{n}}{i}{\mu_{k}^{n}} \geq d_{i}(x_{i}^{o},x_{i}^{n})   \\
\displaystyle\Integral{t^{n}}{t^e}{\overline{\mu}_{i}^{n}}{i}{\overline{\mu}_{k}^{n}} \geq d_{i}(x_{i}^{n},x_{i}^{e}),
  \end{array}\quad i=1,2.
\]
In particular, the following inequalities hold for all $n$:
\[
t^o\leq t^n\leq t^e\quad\hbox{and}\quad\int_{t^{o}}^{t^e}\frac{1}{\sqrt{\alpha_i(s)}}ds \geq d_{i}(x_{i}^{o},x_{i}^{n}),\quad i=1,2.
\]
That is, $t^n \in [t^o,t^e]$ and $x_i^{n} \in \overline{B}_{r_i}(x_i^o)$, $r_i:=\int_{t^{o}}^{t^e}\alpha_i^{-1/2}ds$, $i=1,2$, for large $n$.
But, $[t^o,t^e]$ and $\overline{B}_{r_i}(x_i^o)$, $i=1,2$, are compact sets (recall that $(M_{i},g_{i})$, $i=1,2$, are complete). So,
up to a subsequence, $\{\point{t^n}{x_1^n}{x_2^n}\}_m$ converges to some point $\point{t^*}{x_{1}}{x_{2}}\in V$, which necessarily lies into the (closed) causal diamond $J^+(\point{t^o}{x_1^o}{x_2^o}) \cap J^-(\point{t^e}{x_1^e}{x_2^e})$. In conclusion, the causal diamond is sequentially compact, and so, $(V,g)$ is globally hyperbolic.
\end{proof}



%% file: PartialBoundary.tex
\section{The future c-completion of doubly warped spacetimes}
\label{sec:futurecompletion}
 In this section we are going to study the point set and topological structure of the future c-completion of doubly warped spacetimes.

 Let $\gamma: [\omega,\Omega) \rightarrow V$, $\Omega\leq b$ be a future-directed timelike
curve in $V$. We can reparametrize this curve by using the standard parameter $t$ for the temporal component,
$\gamma(t)=(t,c_{1}(t),c_{2}(t))$. So, from \eqref{eq:3},
\begin{equation}
  \label{eq:4}
  \hbox{length}(c_{i}\mid_{[\omega,\Omega)})
  \leq\int_{\omega}^{\Omega}\frac{ds}{\sqrt{\alpha_{i}(s)}}.
\end{equation}
Next, assume that
$\Omega<b$. Then, the integral in (\ref{eq:4}) is finite. Hence, $\hbox{length}(c_i)<\infty$, and so, $c_{i}(t)\rightarrow x_i^*$ for some $x_i^*\in M_i^C$, where $M_i^C$ denotes the Cauchy completion of the Riemannian manifold $(M_i,g_i)$, $i=1,2$. If, in addition, $x_i^*\in M_i$ for $i=1,2$, the past of $\gamma$ is clearly determined by the triple $(\Omega,x_1^*,x_2^*)$. The following result shows that this is also true if $x_i^*$ belongs to the Cauchy boundary $\partial^C M_i$ for some $i=1,2$.

\begin{prop}\label{pastofcurve}
  Let $\gamma:[\omega,\Omega)\rightarrow V$, $\Omega<b$, be a future-directed timelike curve with $\gamma(t)=(t,c_1(t),c_2(t))$. Then, $\gamma(t)\rightarrow (\Omega, x_1^*,x_2^*)\in (a,b)\times M_1^C\times M_2^C$ for some $(x_1^*,x_2^*)\in M_1^C\times M_2^C$. Moreover, $(t^o,x_1^o,x_2^o)\in I^-(\gamma)$ if, and only if, there exist constants $\mu_{1},\mu_{2}>0$ with $\mu_{1}+\mu_{2}=1$ and
such that
\begin{equation}
  \label{eq:5}
\Integral{t^{o}}{\Omega}{\mu_{i}}{i}{\mu_{k}} >
d_{i}(x^{o}_{i},x^{*}_{i})\qquad\hbox{for $i=1,2$.}
\end{equation}
\end{prop}

\begin{proof}
As argued above, the first assertion is a direct consequence of \eqref{eq:4}. So, we only need to focus on the last assertion.

For the implication to the right, assume that $(t^o,x_1^o,x_2^o)\in I^-(\gamma)$. Since the chronological past $I^-(\gamma)$ is an open set, we can take $\epsilon >0$ small enough so that $(t^o+\epsilon,x_1^o,x_2^o)\in I^-(\gamma)$. Consider an increasing sequence $\{t_n\}\subset [\omega,\Omega)$ with $t_n\nearrow \Omega$ and $(t^o+\epsilon,x_1^o,x_2^o)\ll \gamma(t_n)$ for all $n$. For each $n$, Thm. \ref{c0} ensures the existence of constants $\mu_1^n,\mu_2^n>0$, with $\mu^n_1+\mu^n_2=1$, such that:
  \begin{equation}\label{eq:6}
\Integral{t^{o}+\epsilon}{t_n}{\mu^n_{i}}{i}{\mu^n_{k}} >
d_{i}(x^{o}_{i},c_{i}(t_n))\qquad\hbox{for $i=1,2$.}
    \end{equation}
    Observe that $\{c_i(t_n)\}_n\rightarrow x_i^*\in M_i^C$ for $i=1,2$, and so, from the continuity of the distance function $d_i(x_i^o,\cdot)$ on $M_i^C$, necessarily $\{d_i(x_i^o,x_i(t_n))\}_n\rightarrow d_i(x_i^o,x_i^*)$. Even more, since $\{\mu^n_i\}_n\subset [0,1]$, we can assume that $\{\mu^n_i\}_n$ converges (up to a subsequence) to, say, $\mu_i^*$, $i=1,2$, with $\mu_1^*+\mu_2^*=1$. Hence,
    \[
\left\{\frac{\sqrt{\mu^n_{i}}}{\alpha_{i}(s)}\left(\sum_{k=1}^2 \frac{\mu_{k}^n}{\alpha_{k}(s)}\right)^{-1/2} \right\}_n\rightarrow \frac{\sqrt{\mu^*_{i}}}{\alpha_{i}(s)}\left(\sum_{k=1}^2 \frac{\mu_{k}^*}{\alpha_{k}(s)}\right)^{-1/2}\quad\hbox{pointwise on $[t^o,\Omega]$.}
    \]
    Arguing as in the proof of Thm. \ref{causi}, we observe that these functions are bounded by the integrable function $g:[t^o,\Omega]\rightarrow\R$, $g(t)=\alpha_i(t)^{-1/2}$, so the Dominated Convergence Theorem ensures that
   \[
\left\{\Integral{t^{o}+\epsilon}{t_n}{\mu^n_{i}}{i}{\mu^n_{k}}\right\}_n\rightarrow \Integral{t^{o}+\epsilon}{\Omega}{\mu^*_{i}}{i}{\mu^*_{k}}.
    \]
    In conclusion, by taking limits in \eqref{eq:6}, we arrive to
    \[
\Integral{t^{o}+\epsilon}{\Omega}{\mu^*_{i}}{i}{\mu^*_{k}}\geq d_i(x_i^o,x_i^*)\qquad\hbox{for $i=1,2$.}
      \]
      In order to conclude the implication, it rests to show that, if $t^o+\epsilon$ is replaced by $t^o$, all previous inequalities are strict. In principle, the only way to avoid this conclusion is by assuming that some $\mu_i^*$ is equal to zero. If, say, $\mu_1^*=0$ (and so, $\mu_2^*=1$), then (i) $d_1(x_1^o,x_1^*)=0$ and (ii)
      \[
\int_{t^o}^{\Omega}\frac{1}{\sqrt{\alpha_2(s)}}ds>d_2(x_2^o,x_2^*).
        \]
        Reasoning as in the proof of Prop. \ref{c0}, a small modification of $\mu_1^*,\mu_2^*$ provides new constants $\mu_1,\mu_2>0$, with $\mu_1+\mu_2=1$, such that
\[
          \left\{\begin{array}{ll}\displaystyle
\Integral{t^{o}}{\Omega}{\mu_{1}}{i}{\mu_{k}}>0= d_1(x_1^o,x_1^*)\\
                   \displaystyle\Integral{t^{o}}{\Omega}{\mu_{2}}{i}{\mu_{k}}> d_i(x_2^o,x_2^*),

            \end{array}
          \right.
          \]
          and we are done.

          \smallskip

          For the converse, assume that \eqref{eq:5} holds for some $(t^o,x_1^o,x_2^o)$ and some constants $\mu_1,\mu_2>0$, with $\mu_1+\mu_2=1$, and let us prove that $(t^o,x^o_1,x^o_2)\in I^-(\gamma)$. Recalling that the inequalities in \eqref{eq:5} are strict and $\gamma(t)=(t,c_1(t),c_2(t))\rightarrow (\Omega,x_1^*,x_2^*)$, there exists some $t^e\in (a,b)$ big enough such that
 \[\Integral{t^{o}}{t^e}{\mu_{i}}{i}{\mu_{k}} >
d_{i}(x^{o}_{i},c_{i}(t^e))\qquad\hbox{for $i=1,2$.}
\]
Hence, from Prop. \ref{c0}, $(t^o,x^o_1,x^o_2)\ll \gamma(t^e)$, as required.


\end{proof}

We have just proved that the chronological past of a future-directed timelike curve $\gamma$ defined on a finite interval $[\omega,\Omega)$, $\Omega<b$, is determined by its future limit point $(\Omega,x^*_1,x^*_2)$, in the sense that any other future-directed timelike curve $\gamma'$ with the same future limit point has the same chronological past. Next, we are going to prove that if $\gamma'$ is another future-directed timelike curve converging to another triple, then it generates a different past.

\begin{prop}\label{structuraparcialsininfinito}
  Let $(V,g)$ be a {\multiwarped} spacetime as in (\ref{eq:1-aux}). If $\gamma^i:[\omega^i,\Omega^i)\rightarrow V$, $i=1,2$ satisfy $\gamma^i(t)\rightarrow p_i:=(\Omega^i,x_1^i,x_2^i)\in (a,b)\times M_1^C\times M_2^C$ with $p_1\neq p_2$, then $I^-(\gamma^1)\neq I^-(\gamma^2)$.
\end{prop}
\begin{proof}
    The conclusion easily follows if, say, $\Omega^1<\Omega^2$, since in this case $\gamma^2(t)\in I^-(\gamma^2)\setminus I^-(\gamma^1)$ whenever $\Omega^1<t<\Omega^2$. So, we will assume that $\Omega^1=\Omega^2(=:\Omega)$ and, say, $d_1(x_1^1,x_1^2)>0$. Let $t^o$ be close enough to $\Omega<\infty$ so that (recall that $c_1^1(t)\rightarrow x_1^1$)
 \[
\int_{t^o}^{\Omega}\frac{1}{\sqrt{\alpha_1(s)}}ds<\frac{d_1(x_1^1,x_1^2)}{3}\quad\hbox{and}\quad d_1(c_1^1(t^o),x_1^2)>\frac{d_1(x_1^1,x_1^2)}{3},
    \]
and define $q=\gamma^1(t^o)\in I^-(\gamma^1)$.   
    Then, $q\not\in I^-(\gamma^2)$, since, otherwise, from Prop. \ref{pastofcurve},
    \[
\Integral{t^{o}}{\Omega}{\mu_{1}}{1}{\mu_{k}} >
d_{1}(c^{1}_{1}(t^o),x^{2}_{1})\quad\hbox{for some $\mu_1,\mu_2>0$}.
      \]
      But this is not possible since, from the choice of $t^o$,
      \[
d_1(c^1_1(t^o),x^2_1)>\frac{d_1(x^1_1,x^2_1)}{3}>\int_{t^o}^{\Omega}\frac{1}{\sqrt{\alpha_1(s)}}ds>\Integral{t^{o}}{\Omega}{\mu'_{1}}{1}{\mu'_{k}}      \]
   for any positive constants $\mu'_1,\mu'_2$, with $\mu'_1+\mu'_2=1$. In conclusion, $I^-(\gamma^1)\neq I^-(\gamma^2)$ if $p_1\neq p_2$, and the conclusion follows.
\end{proof}
\begin{rem}\label{rem:1} In the proof of previous result the key property is the finite value of the integral $\int_{t^o}^{\Omega}\alpha_i(s)^{-1/2}ds<\infty$, not the finite value of $\Omega$. Of course, the second imply the first, but the same argument can be reproduced if only the first holds.
%
  \end{rem}
\noindent Props. \ref{pastofcurve} and \ref{structuraparcialsininfinito} together establish a natural bijection between the space $(a,b) \times M_1^C\times M_2^C$ and the set $\hat{V}\setminus \hat{\partial}^{\ncambios{\B}} V$, where $\hat{\partial}^{\ncambios{\B}} V$ denotes the set of TIPs determined by future-directed timelike curves with divergent temporal component ($\Omega=b$). More precisely:

\begin{prop}\label{structuraparcialsininfinito'}
  Let $(V,g)$ be a {\multiwarped} spacetime as in (\ref{eq:1-aux}). Then, there exists a bijection
  \begin{equation}\label{v}
\hat{V}\setminus \hat{\partial}^{\ncambios{\B}} V\; \leftrightarrow\; (a,b)\times M_1^C\times M_2^C,
    \end{equation}
    which maps each indecomposable past set $P\in\hat{V}\setminus \hat{\partial}^{\ncambios{\B}}V$ to the limit point $(\Omega,x^*_1,x^*_2)\in (a,b)\times M_1^C\times M_2^C$ of any future-directed timelike curve generating $P$.
    \end{prop}


\smallskip

Next, we are going to extend the point set structure obtained above to a topological level. We will consider $(a,b)\times M_1^C\times M_2^C$ attached with the product topology. The first result shows the continuity of bijection (\ref{v}) in the left direction:

\begin{prop}\label{prop:topbuenadir}
Let $P_n,P\in \hat{V}\setminus \hat{\partial}^{\ncambios{\B}}V$ with $P_n\equiv (\Omega_n,x_1^n,x_2^n)$ and $P\equiv (\Omega,x_1^*,x_2^*)$, where we are assuming that the triplets belong to $(a,b)\times M_1^C\times M_2^C$. If $(\Omega_n,x_1^n,x_2^n)\rightarrow (\Omega,x_1^*,x_2^*)$, then $P\in \hat{L}(\{P_n\}_n)$.
\end{prop}
\begin{proof}
   First, recall the analytic characterization of the IPs $P$ and $P_n$ provided by Prop. \ref{pastofcurve}: a point $(t,x_1,x_2)\in V$ belongs to $ P$ (resp. $P_n$) if, and only if, there exist positive constants $\mu_1,\mu_2$ ($\mu_1^n,\mu_2^n$) with $\mu_1+\mu_2=1$ ($\mu_1^n+\mu_2^n=1$) and satisfying that, for $i=1,2$:
  \[
    \begin{array}{c}
      \displaystyle\Integral{t}{\Omega}{\mu_{i}}{i}{\mu_{k}}>d_i(x_i,x_i^*) \\
   \left( \displaystyle  \Integral{t}{\Omega_n}{\mu^n_{i}}{i}{\mu^n_{k}}>d_i(x_i,x_i^n)  \right)
    \end{array}
    \]
Second, note that, from the hypotheses, the continuity of the distance map, and the Dominated Convergence Theorem, the following two limits hold: $d_i(x_i,x_i^n)\rightarrow d_i(x_i,x_i^*)$ and
    \[
\Integral{t}{\Omega_n}{\mu_{i}}{i}{\mu_{k}}\rightarrow \Integral{t}{\Omega}{\mu_{i}}{i}{\mu_{k}}\quad\hbox{for any $\mu_1,\mu_2>0$.}
    \]
These two properties directly imply both, $P\subset {\rm LI}(\{P_n\}))$ and $P$ is maximal into ${\rm LS}(\{P_n\})$, i.e., $P\in \hat{L}(\{P_n\})$.
 \end{proof}

In order to prove the continuity of bijection (\ref{v}) in the right direction, we need to impose local compactness on the Cauchy completion, since, otherwise, there exist counterexamples (see, for instance, \cite[Example 4.9]{FHSBuseman})

\begin{prop}\label{topcurvasfinitas}
   Let $(V,g)$ be a doubly warped spacetime as in (\ref{eq:1-aux}) with $M_1^C$ and $ M_2^C$ locally compact. If $\{P_n\}_n$ is a sequence of IPs converging to some IP, $P\equiv (\Omega,x_1^*,x_2^*)\in (a,b) \times M_1^C\times M_2^C$, then $P_n\equiv (\Omega^n,x_1^n,x_2^n)\in (a,b) \times M_1^C\times M_2^C$ for $n$ big enough, and $(\Omega^n,x_1^n,x_2^n)\rightarrow (\Omega,x_1^*,x_2^*)$ with the product topology. As consequence, the bijection (\ref{v}) becomes a homeomorphism.
\end{prop}

\begin{proof}
 The proof follows essentially in the same fashion as \cite[Prop. 5.24]{FHSBuseman}.

  Since the Cauchy completion $M_1^C\times M_2^C$ is locally compact, there exists a pre-compact neighbourhood $U$ of $P\equiv(\Omega,x_1^*,x_2^*)$. Let $\{p^n_m\}_m,\{p_m\}_m\subset V$ be future chains generating $P_n$ and $P$, resp. We can assume without restriction that $\{p_m\}\subset U$.
  It suffices to show the existence of $n_0$ and a map $\mathfrak{m}:\N\rightarrow\N$ such that $p_m^n\in U$ for all $n\geq n_0$ and $m\geq \mathfrak{m}(n)$. In fact, in this case, the temporal component of the sequence $\{p^n_m\}_m$ will not diverge as $m\rightarrow\infty$, and so, $P_n$ can be identified with some $(\Omega^n,x_1^n,x_2^n)\in (a,b) \times M_1^C\times M_2^C\cap\overline{U}$. Moreover, since the result is valid for any pre-compact open set $U$, and $(\Omega,x_1^*,x_2^*)$ admits a countable local neighbourhood basis given by pre-compact open sets, necessarily $(\Omega^n,x_1^n,x_2^n)\rightarrow (\Omega_1^*,x_1^*,x_2^*)$.

In order to prove the statement in previous paragraph, assume by contradiction that, up to subsequences, $p^n_m$ is not contained in $U$ for all $m$ and $n$. Since $P\subset {\mathrm LI}(\{P_n\}_n)$, for each $m\in \N$ there exists $n_0$ such that $p_m\in P_n$ for all $n\geq n_0$. Consider a strictly increasing sequence $\{\mathfrak{n}(m)\}_m$ such that $p_m\in P_{\mathfrak{n}(m)}$. Denote by  $\gamma_m$ the future-directed timelike curve from $p_m$ to some point of a future chain generating $P_{\mathfrak{n}(m)}$. Each $\gamma_m $ intersects the boundary of $\overline{U}$ at a point, say, $(s^m,y_1^m,y_2^m)$. Since $U$ is pre-compact, its boundary is compact and we can assume (up to a subsequence) that $(s^m,y_1^m,y_2^m)\rightarrow (s^*,y^*_1,y^*_2)$ for some $(s^*,y^*_1,y^*_2)\in (a,b) \times M_1^C\times M_2^C$. Let us denote by $P'$ the indecomposable set associated to $q=(s^*,y^*_1,y^*_2)$ which is necesssarily different from $P$ (as $q$ belong to the boundary of $U$); and by $\{q_m\}_m$ a future chain generating $P'$. Next, we are going to show that $P'$ breaks the maximality of $P$ into ${\mathrm LS}(\{P_n\})$, in contradiction with $P\in \hat{L}(\{P_n\}_n)$.

Let us show that $P'\subset {\mathrm LS}(\{P_n\})$. First recall that,  for each $m\in \N$, the set $I^+(q_m)$ is an open set containing $q$:
in fact, this is straightforward if $q\in (a,b)\times M_1\times M_2$; otherwise, it suffices to realize that the characterization of the chronological relation given in Prop. \ref{c0} (which is an open property) extends to the set $(a,b)\times M_1^C\times M_2^C$ (see Prop. \ref{pastofcurve}).
In particular, since $\{(s^k,y_1^k,y_2^k)\}\rightarrow q$, it follows that $(s^k,y_1^k,y_2^k)\in I^+(q_m)$ for $k$ big enough. But, from construction, $(s^k,y_1^k,y_2^k)\in P_{\mathfrak{n}(k)}$, so $q_m\in P_{\mathfrak{n}(k)}$ for $k$ big enough. Therefore, $P'\subset {\mathrm LS}(\{P_{\mathfrak{n}(m)}\}_m)$.

It rests to show that $P\subsetneq P'$; that is, any point $p_m$ of the future chain generating $P$ is contained in $P'$. From construction, $p_m=(t^m,x_1^m,x_2^m)\ll p_k \ll (s^k,y_1^k,y_2^k)$ for all $k> m$. Let $\epsilon>0$ be small enough so that $p_m^\epsilon=(t^m+\epsilon,x_1^m,x_2^m)\ll p_{m+1}$, and thus, $p_m^\epsilon\ll (s^k,y_1^k,y_2^k)$ for all $k>m$. From Prop. \ref{c0}, there exist positive constants $\mu_1^k$ and $\mu_2^k$, with $\mu_1^k+\mu_2^k=1$, such that:
\[
\int_{t^m+\epsilon}^{s^k} \frac{\sqrt{\mu^k_i}}{\alpha_{i}(s)}\left(\sum_{l=1}^{2} \frac{\mu^k_l}{\alpha_{k}(s)} \right)^{-1/2}ds>
d_{i}(x^{m}_{i},y^{k}_{i})\qquad\hbox{for $i=1,2$.}
  \]
 But $\{(s^k,y^k_1,y^k_2)\}\rightarrow (s^*,y_1^*,y_2^*)$. By continuity, and up to a subsequence, there exist positive constants $\mu_1^*,\mu_2^*$, with $\mu_1^*+\mu_2^*=1$, such that:
  \[
\int_{t^m+\epsilon}^{s^*} \frac{\sqrt{\mu^*_i}}{\alpha_{i}(s)}\left(\sum_{l=1}^{2} \frac{\mu^*_l}{\alpha_{k}(s)} \right)^{-1/2}ds \geq
d_{i}(x^{m}_{i},y^{*}_{i})\qquad\hbox{for $i=1,2$.}
    \]
    Now if we replace in previous expression $t^m+\epsilon$ by $t^m$, at least one of previous inequalities becomes strict. Then, reasoning as in the proof of Prop. \ref{c0}, we arrive to

      \[
\int_{t^m}^{s^*} \frac{\sqrt{\mu'_i}}{\alpha_{i}(s)}\left(\sum_{l=1}^{2} \frac{\mu'_l}{\alpha_{k}(s)} \right)^{-1/2}ds >
d_{i}(x^{m}_{i},y^{*}_{i})\qquad\hbox{for $i=1,2$,}
      \]
      for some slightly modified constants $\mu'_i$ from $\mu^*_i$. Therefore, the point $p_m$ belongs to $P'$ (recall Prop. \ref{pastofcurve}). Since this argument works for any point of the sequence $\{p_m\}_m$ generating $P$, the inclusion $P\subsetneq P'$ follows.

    \smallskip

For the last assertion, observe that previous argument gives the continuity of bijection \eqref{v} to the right direction, while Prop. \ref{prop:topbuenadir} ensures the continuity to the left one.

  \end{proof}

 Next, we analyze the case $\Omega=b$. In this case, the finiteness/infiniteness of the integrals associated to the warping functions becomes crucial, so we will consider several subsections to discuss it.

\subsection{Finite warping integrals}

First, we consider the case when the integrals associated to the warping functions are both finite:
\begin{equation}
  \label{eq:7}
  \int_{\C}^{b}\frac{1}{\sqrt{\alpha_i(s)}}ds<\infty, \qquad \hbox{$i=1,2$}\quad\hbox{for some $\C\in (a,b)$.}
\end{equation}
In this case, the following result provides the point set and topological structure of the future c-boundary:
\begin{thm}\label{futurestructurefiniteconditions}
  Let $(V,g)$ be a {\multiwarped} spacetime as in (\ref{eq:1-aux}), and assume that the integral conditions \eqref{eq:7} hold. Then, there exists a bijection
  \begin{equation}
    \label{eq:8}
    \hat{V}\; \leftrightarrow \; (a,b]\times M_1^C\times M_2^C
  \end{equation}
  which maps each IP $P\in \hat{V}$ to the limit point $(\Omega,x_1,x_2)\in (a,b]\times M_1^C\times M_2^C$ of any future-directed timelike curve generating $P$. Moreover, if $M_1^C,M_2^C$ are locally compact, this bijection becomes an homeomorphism.
\end{thm}
\begin{proof}
For the first assertion, we only need to prove the corresponding bijection between $\hat{\partial}^{\ncambios{\B}} V$ and $\{b\}\times M_1^C\times M_2^C$ (recall Prop. \ref{structuraparcialsininfinito'}). But this follows from the same arguments as in the proofs of Props. \ref{pastofcurve} and \ref{structuraparcialsininfinito} (recall \eqref{eq:7} and Remark \ref{rem:1}).
%
%
%
%

For the second assertion, the continuity to the left of bijection (\ref{eq:8}) follows as in Prop. \ref{prop:topbuenadir}, just taking into account that the integral condition \eqref{eq:7} must be used in order to apply the Dominated Convergence Theorem. For the continuity to the right, assume that $P\in \hat{L}(\{P_n\}_n)$, with $P=I^-(\gamma)$, $P_n=I^-(\gamma_n)$, and being $\gamma:[\omega,\Omega)\rightarrow V$, $\gamma_n:[\omega_n,\Omega_n)\rightarrow V$ future-directed timelike curves. Let $(\Omega,x_1^*,x_2^*)$ and $(\Omega_n,x_1^n,x_2^n)$ be the limit points in $(a,b]\times M_1^C\times M_2^C$ of $\gamma$ and $\gamma_n$, resp. We need to prove that $(\Omega_n,x_1^n,x_2^n)\rightarrow (\Omega,x_1^*,x_2^*)$. Observe that, if $\Omega<b$, then the result follows from Prop. \ref{topcurvasfinitas}, so we will focus on the case $\Omega=b$.

First, note that $\Omega_n\rightarrow b$. In fact, otherwise, there exists $\Omega^*<b$ and a subsequence $\{\Omega_{n_k}\}$ such that $\Omega_{n_k}<\Omega^*$ for all $k\in \N$. But, in this case, $P_{n_k}$ will not contain any point $\gamma(t)$ with $t>\Omega^*$, and so, $P\not\subset  {\mathrm LI}(\{P_n\})$.

Assume by contradiction that, say, $\{x_1^n\}_n$ does not converge to $ x_1^*$. Then, up to a subsequence, there exists $\epsilon_0>0$ such that $d_1(x_1^n,x_1^*)>\epsilon_0$. Take $t^0$ big enough so that

  \[
\int_{t^o}^{b}\frac{1}{\sqrt{\alpha_1(s)}}ds<\frac{\epsilon_0}{3}.
    \]
    Take $(x_1^o,x_2^o)\in M_1\times M_2$ such that $q=(t^o,x_1^o,x_2^o)\in I^-(\gamma)=P$ with $d_1(x_1^o,x_1^*)<\epsilon_0/3$. It suffices to show that $q$ does not belong to $P_n$ for any $n$, since, in this case, we arrive to a contradiction with $P\subset {\mathrm LI}(P_n)$.  So, assume that $q\in P_n$ for all $n$. From  Prop. \ref{pastofcurve}, there exists some $\mu^n_1,\mu^n_2>0$ such that
    \[
\Integral{t^{o}}{\Omega^n}{\mu^n_{1}}{1}{\mu^n_{k}} >
d_{1}(x^{o}_{1},x^{n}_{1}).
      \]
      This is in contradiction with the fact that, for any pair of positive constants $\mu'_1,\mu_2'>0$ with $\mu_1'+\mu'_2=1$,
      \[
d_1(x^o_1,x^n_1)>\frac{2}{3} d_1(x^*_1,x^n_1)> \frac{2}{3}\epsilon_0>\int_{t^o}^{b}\frac{1}{\sqrt{\alpha_1(s)}}ds>\Integral{t^{o}}{b}{\mu'_{1}}{1}{\mu'_{k}}.
        \]
\end{proof}

\subsection{One infinite warping integral}

Let us consider now the case when just one of the warping integrals is infinite, say:
\begin{equation}
  \label{eq:9}
 \int_{\C}^{b}\frac{1}{\sqrt{\alpha_1(s)}}ds<\infty \qquad \hbox{and}\qquad \int_{\C}^{b}\frac{1}{\sqrt{\alpha_2(s)}}ds=\infty.
\end{equation}

\subsubsection{Point set structure}

The first integral in condition \eqref{eq:9} plus \eqref{eq:4} ensures that any future-directed timelike curve $\gamma:[\omega,b)\rightarrow V$,  $\gamma(t)=(t,c_1(t),c_2(t))$, satisfies that $c_1(t)\rightarrow x_1^*\in M_1^C$. Moreover, the second integral ensures that the associated Generalized Robertson-Walker spacetime $((a,b) \times M_2,-dt^2+\alpha_2g_2)$ corresponds with the model studied in Section \ref{sec:Robertson}. In particular, since the curve $\sigma(t)=(t,c_2(t))$ is also a future-directed timelike in that spacetime, we can consider the Busemann function $b_{c_2}\in B(M_2)\cup \{\infty\}$.

Next, our aim is to show that the chronological past of $\gamma$ is determined by both, $x_1^*\in M_1^C$ and the Busemann function $b_{c_2}\in B(M_2)\cup \{\infty\}$. Let us begin with the following result:
\begin{prop}\label{prop:conddiferbordedif}
  Let $(V,g)$ be a {\multiwarped} spacetime and assume that the integral conditions \eqref{eq:9} are satisfied. Consider two future-directed timelike curves $\gamma^i:[\omega,b)\rightarrow V$, $\gamma^i(t)=(t,c_1^i(t),c_2^i(t))$, with $c_1^i(t)\rightarrow x_1^i\in M_i^C$, $i=1,2$. If $(x_1^1,b_{c_1})\neq (x_1^2,b_{c_2})$ then $I^-(\gamma^1)\neq I^-(\gamma^2)$.
\end{prop}
\begin{proof} If $x_1^1\neq x_1^2$, we can reason as in the proof of Prop. \ref{structuraparcialsininfinito} (taking $\Omega=\Omega'=b$ and $x_1^1\neq x_1^2$; Remark \ref{rem:1} and the first integral condition in \eqref{eq:9}). So, it suffices to consider the case $b_{c_2^1}\neq b_{c_2^2}$.

Let $\sigma_i(t)=(t,c_2^i(t))$, $i=1,2$, be two future-directed timelike curves on the Generalized Robertson-Walker spacetime $\left( (a,b)\times M_2,-dt^2+\alpha_2 g_2\right)$. Since $b_{c_2^1}\neq b_{c_2^2}$, necessarily $I^{-}(\sigma_1)\neq I^{-}(\sigma_2)$. Assume, for instance, that  $(t^0,y_2)\in I^{-}(\sigma_2)\setminus I^{-}(\sigma_1)$ (the other case is analogous). Then, taking into account the characterization in \eqref{eq:27}, it follows that

  \begin{equation}\label{eq:b}
b_{c_2^1}(y_2)<\int_{\C}^{t^o}\frac{1}{\sqrt{\alpha_2(s)}}ds <b_{c_2^2}(y_2).
\end{equation}
From the first inequality in (\ref{eq:b}),
    \[
      \begin{array}{l}
  (b_{c^1_2}(y_2)=)\lim_{t\rightarrow b} \left(\int_\C^{t}\frac{1}{\sqrt{\alpha_2(s)}}ds-d_2(y_2,c_2^1(t))\right)\leq\int_\C^{t^o}\frac{1}{\sqrt{\alpha_2(s)}}ds\Rightarrow\\ \Rightarrow  \lim_{t\rightarrow b}\left( \int_{t^o}^{t}\frac{1}{\sqrt{\alpha_2(s)}}ds-d_2(y_2,c_2^1(t))\right)\leq 0.
     \end{array}
 \]
 Therefore, since the function $t\mapsto \left(\int_{t^o}^{t}\frac{1}{\sqrt{\alpha_2(s)}}ds-d_2(y_2,c_2^1(t))\right)$ is increasing, we deduce
      \begin{equation}\label{b}
\int_{t^o}^{t}\frac{1}{\sqrt{\alpha_2(s)}}ds<d_2(y_2,c_2^1(t))\quad\hbox{for all $t$.}
        \end{equation}

Let us show the existence of $x_1^o\in M_1$ such that $q=(t^o,x_1^o,y_2)\in I^-(\gamma^2)$. From the inequality
        \[
\int_\C^{t^o}\frac{1}{\sqrt{\alpha_2(s)}}ds <b_{c_2^2}(y_2)=  \lim_{t\rightarrow b}\left(\int_\C^t\frac{1}{\sqrt{\alpha_2(s)}}ds-d_2(y_2,c_2^2(t))\right),
          \]
        there exists $t'>t^o$ big enough such that
        \[
          \int_{t^o}^{t'}\frac{1}{\sqrt{\alpha_2(s)}}ds> d_2(y_2,c_2^2(t')).
        \]
        From continuity, we can find positive constants $\mu_1,\mu_2$, with $\mu_1+\mu_2=1$, such that
        \[
\left\{
  \begin{array}{l}
    \displaystyle\Integral{t^0}{t'}{\mu_2}{2}{\mu_k}>d_2(y_2,c_2^2(t'))\\
    \displaystyle \Integral{t^0}{t'}{\mu_1}{1}{\mu_k}>0.
  \end{array}
\right.
          \]
          So, if we take $x_1^o$ close enough to $c^2_1(t')$ so that
          \[d_1(x_1^o,c^2_1(t'))<\Integral{t^o}{t'}{\mu_1}{1}{\mu_k},\] Prop. \ref{c0} ensures that $(t^o,x_1^o,y_2)\ll \gamma^2(t')$, and thus, $q=(t^o,x_1^o,y_2)\in I^-(\gamma^2)$.

          \smallskip

          On the other hand, for any pair of positive constants $\mu_1,\mu_2>0$ with $\mu_1+\mu_2=1$, necessarily
        \[
\Integral{t^{o}}{t}{\mu_{2}}{2}{\mu_{k}}<\int_{t^o}^t \frac{1}{\sqrt{\alpha_2(s)}}ds< d_2(y_2,c_2^1(t))\quad\hbox{for all $t>t^0$,}
          \]
          where (\ref{b}) has been used in the last inequality. Therefore, from Prop. \ref{c0}, $q\not\ll \gamma^1(t)$ for all $t>t^o$, and thus, $q\not\in I^-(\gamma^1)$.
\end{proof}


\begin{lemma}\label{lemma:aux3}
Let $\gamma:[\omega,\Omega)\rightarrow V$, $\gamma(t)=(t,c_1(t),c_2(t))$ be a future-directed timelike curve with $c_1(t)\rightarrow x_1^*\in M_1^C$. If $\sigma=\{(t_n,x_1^n,c_2(t_n))\}_n\subset V$ satisfies $\{t_n\}_n\rightarrow \Omega$ and $x_1^n\rightarrow x_1^*$, then $I^-(\gamma)\subset {\mathrm LI}(\{I^-(t_n,x_1^n,c_2(t_n))\}_n)$.
\end{lemma}
\begin{proof}
Assume by contradiction the existence of some point $q=(t^o,x_1^o,x_2^o)\in I^-(\gamma)$ such that $q\not\ll (t_n,x^n_1,c_2(t_n))$ for infinitely many $n$. From the open character of the chronological relation, we can assume that $x_1^o\neq x_1^*$. Moreover, for $\epsilon>0$ small enough, it follows that $q_{\epsilon}=(t^o+\epsilon,x_1^o,x_2^o)\in I^{-}(\gamma)$.

Assume that, up to a subsequence, $q_{\epsilon}\ll \gamma(t_n)$ for all $n$. From Prop. \ref{c0}, there exist positive constants $\mu_1^n,\mu_2^n>0$, with $\mu_1^n+\mu_2^n=1$, such that
  \[
\Integral{t^{o}+{\epsilon}}{t_n}{\mu^n_{i}}{i}{\mu^n_{k}} >
d_{i}(x^{o}_{i},c_{i}(t_n))\qquad\hbox{for $i=1,2$.}
    \]
We can assume without restriction that $\{\mu_i^n\}_n$ converges to some point $\mu_i^*$, $i=1,2$. Since $q_{\epsilon}\not\in I^-((t_n,x^n_1,c_2(t_n)))$ (recall that $q\not\ll (t_n,x_1^n,c_2(t_n))$), necessarily
   \[
\left(d_{1}(x^{o}_{1},c_{1}(t_n))< \right)\Integral{t^{o}+{\epsilon}}{t_n}{\mu^n_{1}}{1}{\mu^n_{k}}\leq d_1(x_1^o,x^n_1),
      \]
      the last inequality by Prop. \ref{c0}. From the hypothesis, the first and third element in previous expression converge to $d_1(x_1^o,x_1^*)>0$. Moreover, from \eqref{eq:9}, the integral in the middle is also finite. Hence,
      \begin{equation}\label{eq:c}
\left\{\Integral{t^{o}+\epsilon}{t_n}{\mu^n_{1}}{1}{\mu^n_{k}}\right\}_n\rightarrow \Integral{t^{o}+\epsilon}{\Omega}{\mu^*_{1}}{1}{\mu^*_{k}}=d_1(x_1^o,x_1^*)<\infty.
        \end{equation}
          In particular, since $x_1^o\neq x_1^*$, necessarily $\mu_1^*\neq 0$, and so,
        \begin{equation}\label{eq:cc}
\Integral{t^{o}}{\Omega}{\mu^*_{1}}{1}{\mu^*_{k}}>d_1(x_1^o,x_1^*).
          \end{equation}
          Finally, from (\ref{eq:c}) and (\ref{eq:cc}),
        \[
\int_{t^o}^{t^n}\frac{\sqrt{\mu_1^n}}{\alpha_1}\left(\sum_{k=1}^{2}\frac{\mu_k^n}{\alpha_k}\right)^{-1/2}ds>d_1(x_1^o,x^n_1)\quad\hbox{for $n$ big enough,}
          \]
which implies that $q=(t^o,x_1^o,x_2^o)\in I^-((t_n,x^n_1,c_2(t_n)))$ for $n$ big enough, a contradiction.
\end{proof}

\noindent This Lemma has the following direct consequence:

\begin{lemma}\label{lemma:aux1}
 Let $\gamma^i:[\omega,b)\rightarrow V$, $\gamma^i(t)=(t,c_1^i(t),c_2(t))$, $i=1,2$, be future-directed timelike curves. If $c_1^i(t)\rightarrow x_1^*\in M^C_1$, $i=1,2$, then $I^-(\gamma^1)=I^-(\gamma^2)$.
\end{lemma}
\begin{proof}
 Let us focus on the inclusion to the right (the other one is analogous). Consider the sequence $\sigma=\{(t_n,c_1^2(t_n),c_2(t_n))\}_n$, where $\{t_n\}_n\nearrow \infty$. For any $p\in I^-(\gamma^1)$, Lemma \ref{lemma:aux3} ensures the existence of $n_0$ such that $p\in I^-((t_n,c_1^2(t_n),c_2(t_n)))\subset I^-(\gamma^2)$ for all $n\geq n_0$, as desired.
\end{proof}


\begin{lemma}\label{lemma:aux2}

Let $\gamma^i:[\omega,b)\rightarrow V$, $\gamma^i(t)=(t,c_1(t),c_2^i(t))$, $i=1,2$, be future-directed timelike curves. If $b_{c_2^1}=b_{c_2^2}$, then $I^-(\gamma^1)=I^-(\gamma^2)$.
\end{lemma}
\begin{proof}
Since the first warping integral is finite (recall (\ref{eq:9})), the spatial component $c_1$ admits some limit point $x_1^*\in M_1^C$. Assume by contradiction that, say, $q=(t^o,x_1^o,x_2^o)\in I^-(\gamma^2)\setminus I^-(\gamma^1)$. It is not a restriction to additionally assume that $x_1^o\neq x_1^*$. Let $\epsilon>0$ small enough such that $q_\epsilon=(t^o+\epsilon,x_1^o,x_2^o)\in I^-(\gamma^2)\setminus I^-(\gamma^1)$. Since $q_\epsilon\in I^-(\gamma^2)$, there exists an increasing sequence $\{t_n\}\nearrow b$ such that $q_\epsilon\ll \gamma^2(t_n)$ for all $n$. Then, from  Prop. \ref{c0}, there exist positive constants $\mu_1^n, \mu_2^n >0$, with $\mu_1^n+\mu_2^n=1$, for each $n$, such that
  \begin{equation}\label{eq*}
    \left\{\begin{array}{l}
    \displaystyle  \Integral{t^o+\epsilon}{t_n}{\mu^n_{1}}{1}{\mu^n_{k}}>
             d_{1}(x^o_{1},c_{1}(t_n))\\
\displaystyle\Integral{t^o+\epsilon}{t_n}{\mu^n_{2}}{2}{\mu^n_{k}}>
             d_{2}(x^o_{2},c^2_{2}(t_n)).
    \end{array}\right.
    \end{equation}
    It is not a restriction to assume that each sequence $\{\mu^n_i\}_n$ is convergent to $\mu_i^*$ for $i=1,2$. Next, we are going to prove that the sequences can be chosen satisfying $\mu_{1}^{*} \neq 1$:

   \smallskip

{\em Claim. The sequences $\{\mu^n_i\}_n$  can be chosen so that $\mu_1^*\neq 1$ (and thus, $\mu_2^*\neq 0$).}

  \smallskip

{\em Proof of the Claim.} Let us prove that, if we have a sequence $\{t_n\}_n$ such that $q=(t^o,x_1^o,x_2^o)\ll \gamma(t_n)=(t_n,c_1(t_n),c_2(t_n)) $, we can find sequences $\{\mu_i^n\}_n$, with $\mu_1^n+\mu_2^n=1$, which converge, up to a subsequence, to $\mu_1^*\neq 1$ and $\mu_2^*\neq 0$, such that
\begin{equation}
  \label{eq:30}
    \Integral{t^o}{t_n}{\mu^n_i}{i}{\mu^n_{k}}-
             d_{i}(x^o_{i},c_{i}(t_n))>0\quad \hbox{for $i=1,2$.}
\end{equation}
Observe that Prop. \ref{c0} ensures the existence of such convergent sequences $\{\mu_i^n\}_n$  without the statement about the limits. Assume that $\mu_1^*= 1$. By using standard arguments (that is, working with the point $q_{\epsilon}=(t^o+\epsilon,x_1^o,x_2^o)$ as before, and recalling that $\mu_1^n\geq \frac{1}{2}$ for $n$ big enough), we can take limits on \eqref{eq:30} preserving the strict inequality. So,
\[
\lim_{n\rightarrow \infty}\left(\Integral{t^o}{t_n}{\mu^n_1}{1}{\mu^n_{k}}-
             d_{1}(x^o_{1},c_{1}(t_n))\right)=\int_{t^o}^{b}\frac{1}{\sqrt{\alpha_1(s)}}ds-
             d_{1}(x^o_{1},x_1^*)>0
           \]
where $c_1(t_n)\rightarrow x_1^*$. Now take $\overline{\mu}_2^*>0$ small enough such that $\overline{\mu}_1^*=1-\overline{\mu}_2^*>0$ and such that
           \[
\Integral{t^o}{b}{\overline{\mu}_1^*}{1}{\overline{\mu}^*_k}-d_{1}(x^o_{1},x_1^*)>0
             \]
Now, define $\overline{\mu}_1^n=\mu_1^n-\overline{\mu}_2^*$ and $\overline{\mu}_2^n=\mu_2^n+\overline{\mu}_2^*$. As $\mu_{1}^{n} \rightarrow 1$,  we have that $\overline{\mu}_{1}^{n}>0$ for large $n$ and that
$\overline{\mu}_{1}^{n} \rightarrow \overline{\mu_{1}^{*}}$, therefore by  the Dominated Convergence Theorem (recall the integral condition for $\alpha_1$) we have:
\[
\Integral{t^o}{b}{\overline{\mu}_1^*}{1}{\overline{\mu}^*_k}=lim_{n} \Integral{t^o}{t_{n}}{\overline{\mu}_1^n}{1}{\overline{\mu}^n_k},
\]
Hence,
\[
  lim_{n}\left(\Integral{t^o}{t_{n}}{\overline{\mu}_1^n}{1}{\overline{\mu}^n_k}-d_{1}(x_{1}^{o},c_{1}(t_{n}))\right)
  >0,
\]
and so  for large $n$ we can take $\overline{\mu}_1^n$ and $\overline{\mu}_2^n$ satisfying
\[
\Integral{t^o}{t_n}{\overline{\mu}^n_1}{1}{\overline{\mu}^n_{k}}-
             d_{1}(x^o_{1},c_{1}(t_n))>0.
  \]
  Moreover, as $\overline{\mu}_1^n<\mu_1^n$ and $\overline{\mu}_2^n>\mu_2^n$, it easily follows that:

  \[
\Integral{t^o}{t_n}{\overline{\mu}^n_2}{2}{\overline{\mu}^n_{k}}>\Integral{t^o}{t_n}{\mu^n_2}{2}{\mu^n_{k}}\left(>d_2(x_2^0,c_2(t_n))\right).
    \]
In conclusion, equation \eqref{eq:30} is also true with the sequences $\{\overline{\mu}_i^n\}_n$ and $\{\overline{\mu}_1^n\}_n\rightarrow \overline{\mu}_{1}^{*}= 1-\overline{\mu}_{2}^{*}\neq 1$, which proves the claim.


    \smallskip

  Continuing with the proof of the lemma, note that $\gamma^1$ and $\gamma^2$ share the same first spatial component $c_1$, the first integral condition (\ref{eq*}) coincides for both curves. Therefore, since $q_\epsilon\not\in I^-(\gamma^1)$, necessarily (recall Prop. \ref{c0}):
    \begin{equation}
      \label{eq:11}
d_2(x^o_2,c^1_2(t_n))\geq \Integral{t^o+\epsilon}{t_n}{\mu^n_{2}}{2}{\mu^n_{k}}\left(>
             d_{2}(x^o_{2},c^2_{2}(t_n))\right).
    \end{equation}
    Moreover, from the hypothesis, $b_{c_2^1}(x_2^o)=b_{c_2^2}(x_2^o)$. So, from the definition of Busemann function,
    \begin{equation}\label{x}
    \lim_{n}\left(d_2(x_2^o,c_2^1(t_n))-d_2(x_2^o,c_2^2(t_n))\right)=0.
    \end{equation}
    From \eqref{eq:11} and (\ref{x})
    \begin{equation}
      \label{eq:12}
\lim_n \left(\Integral{t^o+\epsilon}{t_n}{\mu^n_{2}}{2}{\mu^n_{k}}-d_2(x_2^o,c_2^1(t_n))\right)=0.
    \end{equation}

On the other hand, from the claim, the sequence of positive constants $\{\mu_2^n\}_n$ does not converge to $0$, so there exists ${\cal K}>0$ such that
   \begin{equation}
     \label{eq:13}
     \Integral{t^o}{t^o+\epsilon}{\mu_2^n}{2}{\mu_k^n}>{\cal K}>0\quad\hbox{for $n$ big enough.}
   \end{equation}
        So, putting together \eqref{eq:12} and \eqref{eq:13} we obtain that
\[
\lim_n \left(\Integral{t^o}{t_n}{\mu^n_{2}}{2}{\mu^n_{k}}-d_2(x_2^o,c_2^1(t_n))\right)>0,
  \]
  and thus,
  \[
\Integral{t^o}{t_n}{\mu^n_{2}}{2}{\mu^n_{k}}>d_2(x_2^o,c_2^1(t_n))\quad\hbox{for $n$ big enough.}
    \]
    From Prop. \ref{c0}, necessarily $q\in I^-(\gamma^1)$, a contradiction.

\end{proof}
\noindent As a direct consequence of Lemmas \ref{lemma:aux1} and \ref{lemma:aux2}, we obtain:
\begin{prop}\label{samecondsamepast}
Let $\gamma^i:[\omega,b)\rightarrow V$, $\gamma^i(t)=(t,c_1^i(t),c_2^i(t))$, $i=1,2$, be future-directed timelike curves. If $c_1^i(t)\rightarrow x_1^*\in M_1^C$, $i=1,2$, and $b_{c_2^1}=b_{c_2^2}$, then $I^-(\gamma^1)=I^-(\gamma^2)$.
\end{prop}
\begin{proof}
Let $c_1:[\omega,b)\rightarrow M_1$ be a curve with $c_1(t)\rightarrow x_1^*$ such that the curves $\overline{\gamma}^i:[\omega,\infty)\rightarrow V$, $\overline{\gamma}^i(t)=(t,c_1(t),c_2^i(t))$, $i=1,2$, are future-directed timelike. From Lemma \ref{lemma:aux1}, $I^-(\gamma^i)=I^-(\overline{\gamma}^i)$, $i=1,2$. But $\overline{\gamma}^1$, $\overline{\gamma}^2$ share the same first spatial components, and their second spatial components define the same Busemann function. Hence, Lemma \ref{lemma:aux2} ensures that $I^-(\overline{\gamma}^1)=I^-(\overline{\gamma}^2)$, as required.
\end{proof}

%

\noindent Summarizing, if we put together Props. \ref{structuraparcialsininfinito'}, \ref{prop:conddiferbordedif}
and \ref{samecondsamepast}, we deduce the following point set structure for the future c-completion of $(V,g)$:
\begin{thm}\label{futurecomploneinfinite}
  Let $(V,g)$ be a  {\multiwarped} spacetime as in \eqref{eq:1-aux}, and assume that the integral conditions \eqref{eq:9} hold. Then, there exists a bijection
 \begin{equation}
   \label{eq:10}
     \hat{V}\; \leftrightarrow \;  M_1^C\times \left(B(M_2)\cup \{\infty\}\right)\;\equiv\;
       \left( (a,b)\times M_1^C\times M_2^C\right) \cup M_{1}^{C} \times \left({\cal B}(M_2)\cup \{\infty\}\right).
%
  \end{equation}
 This bijection maps each indecomposable past set $P=I^-(\gamma)\in \hat{V}$, where $\gamma:[\omega,\Omega)\rightarrow V$, $\gamma(t)=(t,c_1(t),c_2(t))$, is any curve generating $P$, to a pair $(x_1^*,b_{c_2})$, where $x_1^*\in M_1^C$ is the limit point of the curve $c_1$. If $\Omega<b$, then $b_{c_2}=d_{(\Omega,x_2^*)}$, where $x_2^*$ is the limit point of $c_2$ (see \eqref{eq:46}), and thus, $P$ can be also identified with the limit point $(\Omega,x_1^*,x_2^*)$ of $\gamma$ (recall Prop. \ref{structuraparcialsininfinito'}).

\end{thm}


\subsubsection{Topological Structure}

 Next, we are going to extend previous study to a topological level, showing that the bijection obtained above is actually a homeomorphism when the corresponding product topology on $M_1^C\times (B(M_2)\cup \{\infty\})$ is considered.

 To this aim, we only need to prove the following equivalence:
given $P\equiv (x_1^*,b_{c_2})\in \hat{V}$ and $\{P_n\}_n\equiv \{(x_1^n,b_{c_2^n})\}_n\subset \hat{V}$,
\begin{equation}\label{equ}
P\in \hat{L}(\{P_n\}_n)\iff  x_1^n\rightarrow x^*_1\;\;\hbox{and}\;\; b_{c_2}\in \hat{L}(\{b_{c_2^n}\}_n).
  \end{equation}
  Under the hypothesis of $M_1^C$ and $M_2^C$ being locally compact, the equivalence (\ref{equ}) for the case $b_{c_2}\equiv d_{(\Omega,x_2)}$ is already proved in Prop. \ref{topcurvasfinitas}. In fact, if $P_n=I^-(\gamma^n)$ with $\gamma^n:[\omega,\Omega_n)\rightarrow V$, then $\Omega_n<b$ for $n$ big enough. In particular, $b_{c_2^n}\equiv d_{(\Omega_n,x^n_2)}$ with $x_2^n\in M_2^C$ (see \eqref{eq:46}). Moreover, Prop. \ref{topcurvasfinitas} implies that $(\Omega_n,x_1^n,x_2^n)\rightarrow (\Omega,x^*_1,x^*_2)$. Hence, $\{d_{(\Omega_n,x_2^n)}\}_n$ converges pointwise to $d_{(\Omega,x_2)}$, and thus, $d_{(\Omega,x_2)}\in \hat{L}(\{d_{(\Omega_n,x_2^n)}\}_n)$ (see Prop. \ref{prel:PropToponefibre}). So, to finish the proof of (\ref{equ}), we can focus just on the case $b_{c_2}\in {\cal B}(M_2)$.

  We begin with some preliminary results.
  \begin{lemma}\label{lemma:aux5}
    Let $P,P'\in \hat{V}$ and $\{P_n\}_n\subset \hat{V}$, and assume that $P\equiv (x_1,b_{c_2}), P'\equiv (x_1',b_{c'_2})$ and $P_n\equiv (x_1^n,b_{c_2^n})$ belong to $M_1^C\times \left(B(M_2)\cup \{\infty\}\right)$ for all $n$ (recall the identification in \eqref{eq:10}). Then, the following statements hold:
    \begin{itemize}
    \item[(i)] If $x_1=x_1'$, then
        \[
b_{c_2}\leq b_{c_2'} \iff P\subset P'.
        \]
    \item[(ii)] If $x_1^n\rightarrow x_1$, then
      \[
P\subset {\mathrm LI}(\{P_n\}_n) \iff b_{c_2}\leq {\mathrm lim\, inf}_n (\{b_{c_2^n}\}_n).
        \]
    \end{itemize}
  \end{lemma}
  \begin{proof} Let $\gamma:[\omega,\Omega)\rightarrow V$, $\gamma':[\omega',\Omega')\rightarrow V$ and $\gamma^n:[\omega^n,\Omega^n)\rightarrow V$ be future-directed timelike curves generating $P,P'$ and $P_n$, resp.

    (i) First, let us prove the implication to the left. Assume that $\gamma(t)=(t,c_1(t),c_2(t))$ and $\gamma'(t)=(t,c_1'(t),c_2'(t))$ satisfy that $c_{1}(t) \rightarrow x_{1}$, $c_{1}'(t) \rightarrow x_{1}$ and $b_{c_{2}}$, $b_{c_{2}'}$ are their Busemann functions. Consider the future-directed timelike curves $\sigma(t)=(t,c_2(t))$ and $\sigma'(t)=(t,c_2'(t))$ in the Generalized Robertson-Walker spacetime $$\left( (a,b)\times M_2,-dt^2+\alpha_2g_2\right).$$ Since $P\subset P'$, necessarily $P(b_{c_2})=I^-(\sigma)\subset I^-(\sigma')=P(b_{c_2'})$, and thus, $b_{c_2}\leq b_{c_2'}$ (recall \eqref{eq:27} and \eqref{eq:28}).

For the implication to the right, assume that $x_1=x_1'$ and $b_{c_2}\leq b_{c'_2}$. It suffices to show the existence of a sequence $\sigma=(t_n,y_1^n,c_2(t_n))$ with $\{t_n\}_n\nearrow \Omega$, satisfying $\{y_1^n\}_n\rightarrow x_1$ and $\sigma\subset P'$. In fact, in this case, Lemma \ref{lemma:aux3} ensures that $P\subset {\rm LI}(\sigma)$ and, taking into account that $\sigma\subset P'$, necessarily $P\subset P'$.

To this aim, take $\{t_n\}_n\nearrow \Omega$ and observe that, by hypothesis, $b_{c_2}\leq b_{c'_2}$. So, in the Generalized Robertson-Walker spacetime $\left( (a,b)\times M_2,-dt^2+{\alpha_2}g_2\right)$, the inclusion $P(b_{c_2})\subset P(b_{c'_2})$ holds (recall equations \eqref{eq:27} and \eqref{eq:28}). In particular, since the future-directed timelike curves $\sigma(t)=(t,c_2(t))$ and $\sigma'(t)=(t,c'_2(t))$ satisfy $I^-(\sigma)=P(b_{c_2})$ and $ I^-(\sigma')=P(b_{c'_2})$, there exists a sequence $\{s_n\}_n$, with $\{s_n\}_n\nearrow \Omega'$, such that
$\sigma(t_n)=(t_n,c_2(t_n))\ll (s_n,c'_2(s_n))=\sigma'(s_n)$. Let us show that $\sigma=\{(t_n,c_1'(s_n),c_2(t_n))\}$ is the required sequence. From construction and the fact that $(t_n,c_1'(s_n),c_2(t_n))\ll (s_n,c_1'(s_n),c'_2(s_n))$ in $V$ for all $n$, necessarily $\sigma\subset P'$. Moreover, since $\{s_n\}_n\nearrow \Omega'$, necessarily $c_1'(s_n)\rightarrow x_1'=x_1$, as desired.

\smallskip

(ii) For the implication to the right, assume that $P\subset {\mathrm LI}(\{P_n\}_n)$ and let us show that $b_{c_2}\leq \liminf (\{b_{c_2^n}\}_n)$. Denote by $\sigma(t)=(t,c_2(t))$ and $\sigma_n(t)=(t,c_2^n(t))$ future-directed timelike curves in the Generalized Robertson Walker model \[\left((a,b)\times M_2,-dt^2+ \alpha_2g_2\right).\] Since $P\subset {\mathrm LI}(\{P_n\}_n)$, necessarily
  \[
P(b_{c_2})=I^-(\sigma)\subset {\mathrm LI}(\{I^-(\sigma_n)\}_n)={\mathrm LI}(\{P(b_{c_2^n})\}_n)
    \]
(where we are considering past sets in the associated Generalized Robertson Walker model), and the conclusion follows from \eqref{eq:50}.

        \smallskip

        For the implication to the left, assume that $b_{c_2}\leq {\mathrm lim\, inf}_n (\{b_{c_2^n}\}_n)$ and let us prove that $P\subset {\mathrm LI}(\{P_n\}_n)$. Let $\{t_k\}\nearrow \Omega$ be an arbitrary sequence. For each $k$, and from the timelike character of $\gamma$, we have $(t_k,c_2(t_k))\ll (t,c_2(t))$ in the Generalized Robertson-Walker spacetime $\left( (a,b)\times M_2,-dt^2+{\alpha_2}g_2\right)$ for all $t>t_k$. From \eqref{eq:26} and the increasing character of \eqref{eq:25},

      \[
\int_\C^{t_k}\frac{1}{\sqrt{\alpha_2(s)}}ds < b_{c_2}(c_2(t_k))=\lim_{t\rightarrow \Omega} \left(\int_\C^t \frac{1}{\sqrt{\alpha_2(s)}}ds - d_2(c_2(t_k),c_2(t))\right).
        \]
Since $b_{c_2}\leq {\mathrm lim\,inf}(\{b_{c_2^n}\}_n)$, there exists an increasing sequence $\{n_k\}_k$ such that
\begin{equation}
  \label{eq:16}
\int_\C^{t_k}\frac{1}{\sqrt{\alpha_2(s)}}ds<b_{c^n_2}(c_2(t_k))=\lim_{r\rightarrow \Omega_n} \left(\int_\C^{r}\frac{1}{\sqrt{\alpha_2(s)}}ds- d_2(c_2(t_k),c_2^n(r))\right)\quad\hbox{$\forall$ $n\geq n_k$.}
\end{equation}
  For each $n_k\leq n<n_{k+1}$, consider $r_n\in [\omega_{n},\Omega_n)$ such that
  \begin{equation}
    \label{eq:24}
\int_\C^{t_k}\frac{1}{\sqrt{\alpha_2(s)}}ds< \int_\C^{r_n}\frac{1}{\sqrt{\alpha_2(s)}}ds-d_2(c_2(t_k),c_2^n(r_n)),\qquad d_{1}(c_1^n(r_n),x_1^n)<\frac{1}{2^n},
  \end{equation}
    (for the first inequality recall \eqref{eq:16}; for the second one, recall that $c_1^n(t)\rightarrow x_1^n$). From the first inequality, it follows that
    \[
    (t_k,c_2(t_k))\ll (r_n,c_2^n(r_n))\quad\hbox{for $n_k\leq n< n_{k+1}$ and all $k$.}
    \]
    However, since $\{(t_k,c_2(t_k))\}$ is a chronological chain, previous chronological relation is true for all $n\geq n_k$: in fact, if $n\geq n_k$, there exists $k'(\geq k)$ such that $n_{k'}\leq n < n_{k'+1}$. As we have noted before $(t_{k'},c_2(t_{k'}))\ll (r_n,c_2^n(r_n)))$ but, taking into account $(t_k,c_2(t_k))\ll (t_{k'},c_2(t_{k'}))$, necessarily $(t_{k},c_2(t_{k}))\ll (r_n,c_2^n(r_n)))$.

    Next, define the sequence $\sigma=\{(l_n,c^n_1(r_n),c_2(l_n))\}_n$, where $l_n:=t_k$ if $n_{k}\leq n< n_{k+1}$. Since $\{t_k\}_k\nearrow \Omega$, necessarily $\{l_n\}_n\rightarrow \Omega$. Moreover, since $(t_k,c_2(t_k))\ll (r_n,c_2^n(r_n))$,
    \[(l_n,c^n_1(r_n),c_2(l_n))=(t_k,c_1^n(r_n),c_2(t_k))\ll (r_n,c_1^n(r_n),c^n_2(r_n))=\gamma^n(r_n),\] hence $(l_n,c^n_1(r_n),c_2(l_n))\in P_n$ for all $n$. Finally, note that $\sigma$ satisfies the conditions of Lemma \ref{lemma:aux3}, as $\{l_n\}_n\rightarrow \Omega$ and $c_1^n(r_n)\rightarrow x_1$ (recall that $c_1^n(t)\rightarrow x_1^n$, $x_1^n\rightarrow x_1$ from hypothesis and the second inequality in \eqref{eq:24}). Therefore,
\[
    P\subset {\mathrm LI}(\{I^-(l_n,c^n_1(r_n),c_2(l_n))\}_n)\subset {\mathrm LI}(\{P_n\}),
    \]
    as desired.
  \end{proof}


  \begin{prop}\label{prop:topcharac} Let $P\in \hat{V}$ and $\{P_n\}_n\subset \hat{V}$, and assume that $P\equiv (x_1,b_{c_2})$ and $P_n\equiv (x_1^n,b_{c_2^n})$ (in $M_1^C\times \left(B(M_2)\cup \{\infty\} \right)$) for all $n$. Then,
$P\in \hat{L}(\{P_n\}_n)$ if, and only if, $x_1^n\rightarrow x_1$ and $b_{c_2}\in \hat{L}(\{b_{c_2^n}\}_n)$.
  \end{prop}
  \begin{proof}
    For the implication to the right, and reasoning as in the proof of Thm. \ref{futurestructurefiniteconditions}, it follows that $x_1^n\rightarrow x_1$ (recall the finite warping integral in \eqref{eq:9} and Remark \ref{rem:1}). Hence, we will focus on $b_{c_2}\in \hat{L}(\{b_{c_2^n}\})$. From Lemma \ref{lemma:aux5} and the fact that $P\in \hat{L}(\{P_n\}_n)$, necessarily $b_{c_2}\leq \liminf(\{b_ {c_2^n}\}_n)$. So, $b_{c_2}\in \hat{L}(\{b_{c_2^n}\})$ follows if we prove that $b_{c_2}$ is maximal into $\limsup(\{b_{c^n_2}\}_n)$. Consider any $b_{\overline{c}_2}$ such that $b_{c_2}\leq b_{\overline{c}_2}\leq {\mathrm lim\,sup}(\{b_{c^n_2}\}_n)$, and consider the associated past set $\overline{P}\equiv (x_1,b_{\overline{c}_2})$. Up to a subsequence, we can assume that $b_{\overline{c}_2}\leq {\mathrm lim\,inf}(\{b_{c^n_2}\}_n)$. From Lemma \ref{lemma:aux5}, $P\subset \overline{P}$ and $\overline{P}\subset {\mathrm LI}(\{P_{n}\}_n)$. But $P$ is maximal into the superior limit of the sequence $\{P_n\}_n$, so necessarily $P=\overline{P}$. From Prop. \ref{prop:conddiferbordedif} we have that $b_{c_2}=b_{\overline{c}_2}$ so the maximal character of $b_{c_2}$ into $\limsup(\{b_{c_2^n}\}_n)$ is obtained.

    \smallskip

    For the implication to the left, first note that $P\subset {\mathrm LI}(\{P_n\}_n)$ (recall Lemma \ref{lemma:aux5} and the definition of $\hat{L}$ for Busemann functions \eqref{eq:22}). So, we only need to focus on the maximal character of $P$ into ${\mathrm LS}(\{P_n\})$. Take $\overline{P}$ an indecomposable past set with $P\subset \overline{P}$ and maximal into ${\mathrm LS}(\{P_n\})$, and let us prove that $P=\overline{P}$. Assume that $\overline{P}\equiv (\overline{x}_1,b_{\overline{c}_2})$. Up to a subsequence, we can also assume that $\overline{P}\subset {\mathrm LI}(\{P_n\})$, hence $\overline{P}\in \hat{L}(\{P_n\})$. Hence, from previous part, $x_1^n\rightarrow \overline{x}_1$. But, by hypothesis, $x_1^n\rightarrow x_1$, obtaining that $x_1=\overline{x}_1$. Once this is observed, Lemma \ref{lemma:aux5} ensures both, $b_{c_2}\leq b_{\overline{c}_2}$ and $b_{\overline{c}_2}\leq \limsup(\{b_{c_2^n}\})$. Since $b_{c_2}\in \hat{L}(\{b_{c_2^n}\})$, necessarily $b_{c_2}=b_{\overline{c}_2}$, and so, $P=\overline{P}$ (recall Prop. \ref{samecondsamepast}).
  \end{proof}

\noindent Summarizing, we are in conditions to deduce the following result:

\begin{thm}\label{futurecomploneinfinite}
  Let $(V,g)$ be a  {\multiwarped} spacetime as in \eqref{eq:1-aux}, and assume that the integral conditions in \eqref{eq:9} are satisfied. If $M_1^C$ and $M_2^C$ are locally compact, the bijection (\ref{eq:10}) becomes a homeomorphism.
\end{thm}
\begin{proof}
  From Prop. \ref{topcurvasfinitas}, the bijection between $\hat{V}\setminus \hat{\partial}^{\B} V$ and $(a,b) \times M_1^C\times M_2^C$ is a homeomorphism if we assume that $M_1^C$ and $M_2^C$ are locally compact. From Prop. \ref{prop:topcharac}, the homeomorphism can be extended to the bijection (\ref{eq:10}).
\end{proof}

\section{The past c-completion of doubly warped spacetimes}\label{ss6}

Obviously, similar arguments provide the corresponding results for the past c-completion:

\begin{thm}\label{pfuturestructurefiniteconditions}
  Let $(V,g)$ be a {\multiwarped} spacetime as in (\ref{eq:1-aux}), and assume that the integral conditions
  \begin{equation}
  \label{eqq:7}
  \int_{a}^{\C}\frac{1}{\sqrt{\alpha_i(s)}}ds<\infty, \qquad \hbox{$i=1,2$}\quad\hbox{for some $\C\in (a,b)$.}
\end{equation}
 hold. Then, there exists a bijection
  \begin{equation}
    \label{eqq:8}
    \check{V}\; \leftrightarrow \; [a,b) \times M_1^C\times M_2^C
  \end{equation}
  which maps each IF $F\in \check{V}$ to the limit point $(\Omega,x_1,x_2)\in [a,b)\times M_1^C\times M_2^C$ of any past-directed timelike curve generating $F$. Moreover, if $M_1^C$ and $M_2^C$ are locally compact, then this bijection becomes a homeomorphism.
\end{thm}

\begin{thm}\label{pfuturecomploneinfinite}
  Let $(V,g)$ be a  {\multiwarped} spacetime as in \eqref{eq:1-aux}, and assume that the integral conditions
  \begin{equation}
  \label{eqq:9}
 \int_{a}^{\C}\frac{1}{\sqrt{\alpha_1(s)}}ds<\infty \qquad \hbox{and}\qquad \int_{a}^{\C}\frac{1}{\sqrt{\alpha_2(s)}}ds=\infty,
\end{equation}
hold. Then, there exists a bijection
 \begin{equation}
   \label{eqq:10}
     \check{V}\; \leftrightarrow\;  M_1^C\times \left(B(M_2)\cup \{-\infty\}\right)
     \equiv  \left( (a,b)\times M_1^C\times M_2^C\right) \cup M_{1}^{C} \times \left({\cal B}(M_2)\cup \{\infty\}\right).
  \end{equation}
  This bijection maps each indecomposable future set $F=I^+(\gamma)\in \check{V}$, where $\gamma:[\omega,-\Omega)\rightarrow V$, $\gamma(t)=(-t,c_1(t),c_2(t))$, is any curve generating $F$, to a pair $(x_1^*,b^-_{c_2})$, where $x_1^*\in M_1^C$ is the limit point of the curve $c_1$. If $-\Omega>a$, then  $b^-_{c_2}=d^-_{(\Omega,x_2^*)}$, where $x_2^*$ is the limit point of $c_2$ (see \eqref{eq:48}), and thus, $F$ can be also identified with the limit point $(\Omega,x_1^*,x_2^*)$ of $\gamma$.
\end{thm}

%% file: CompleteBoundary.tex
\section{The total c-completion of doubly warped spacetimes}
\label{sec:totalcompletion}

We are now in conditions to construct the (total) c-completion of doubly warped spacetimes by merging appropriately the future and past c-boundaries obtained in previous section.

To this aim, first we need to determine the S-relation between indecomposable sets. So, let $\gamma:[\omega,\Omega)\rightarrow V$, $\gamma(t)=(t,c_{1}(t),c_{2}(t))$, be an inextensible future-directed timelike curve.
%
%
Clearly, if $\Omega=b$ then $\uparrow I^{-}(\gamma)=\emptyset$, and there are no IFs S-related to $I^{-}(\gamma)$. So, we will focus on the case $\Omega<b$.

\begin{prop}
\label{tip}
Let $(V,g)$ be a \multiwarped spacetime and consider a future-directed (resp. past-directed) timelike curve $\gamma$ with associated endpoint $(\Omega^+,x_1^*,x_2^*) \in (a,b) \times M_{1}^{C} \times M_{2}^{C}$  (resp. $(\Omega^-,y_1^*,y_2^*) \in (a,b) \times M_{1}^{C} \times M_{2}^{C}$). Then
\begin{equation}
\begin{aligned}
        \uparrow I^{-}(\gamma) &=\{(t,x_1,x_2) \in V\; \mid\; \exists \,  \mu_{1},\mu_{2} > 0\;\; \hbox{{\rm such that}} \\
        & \Integral{\Omega^+}{t}{\mu_{i}}{i}{\mu_{k}}> d_{i}(x_{i},x_{i}^*),\; i=1,2.\}
\end{aligned}
\end{equation}

\begin{equation*}
\begin{aligned}
(\hbox{resp.}\;\;\downarrow I^{+}(\gamma) &=\{(t,x_1,x_2) \in V\; \mid\; \exists \,  \mu_{1},\mu_{2} > 0\;\; \hbox{{\rm such that}} \\
                & \Integral{t}{\Omega^-}{\mu_{i}}{i}{\mu_{k}}> d_{i}(x_{i},y^{*}_{i}),\; i=1,2\}).
\end{aligned}
\end{equation*}
As consequence, if $P\in\hat{V}$ and $F\in\check{V}$ are associated to $(\Omega^+,x^*_1,x^*_2)$ and $(\Omega^-,y^*_1,y^*_2)$ in $(a,b) \times M_{1}^{C} \times M_{2}^{C}$, resp, then the following equivalence holds:
\[
P \sim_{S} F\quad \Longleftrightarrow\quad
\Omega^{-}=\Omega^{+}\;\;\hbox{and}\;\; x^*_{i}=y^*_{i} \in M_{i}^{C},\; i=1,2.
\]
\end{prop}
\begin{proof} Assume that $\gamma:[\omega,\Omega^+) \rightarrow V$, $\gamma(t)=(t,c_1(t),c_2(t))$, is a future-directed causal curve with associated endpoint $(\Omega^+,x_1^*,x_2^*) \in (a,b) \times M_{1}^{C} \times M_{2}^{C}$ (for the past is analogous). We need to show that $\uparrow I^-(\gamma)=A_{(\Omega,x_1^*,x_2^*)}$, where
\begin{equation*}
\begin{aligned}
A_{(\Omega^+,x_1^{*},x_2^*)}&:=\{(r,x_1,x_2) \in V \mid \exists \,  \mu_{1},\mu_{2} > 0\;\; \hbox{{\rm such that}} \\
        &
        \int_{\Omega^+}^{r}\frac{\sqrt{\mu_{i}}}{\alpha_i(s)}\left(\sum_{k=1}^2\frac{\mu_k}{\alpha_k(s)}\right)^{-1/2}dt > d_{i}(x_{i},x_{i}^*),\; i=1,2\}.
\end{aligned}
\end{equation*}
For the inclusion to the right, take $(r,x_1,x_2) \in \uparrow I^{-}(\gamma)$ and $\epsilon>0$ small enough so that $(r-\epsilon,x_1,x_2)\in \uparrow I^-(\gamma)$ (recall that the common future is open). For any sequence $\{t_{n}\}_{n}\nearrow \Omega^+$ we have $\gamma(t_n) \ll (r-\epsilon,x_1,x_2)$ for all $n$. From Prop. \ref{c0} there exist constants $\mu^n_{1},\mu^n_{2}>0$, with $\mu^n_{1}+\mu_{2}^n=1$ for all $n$, such that
\[
\int_{t_n}^{r-\epsilon}\frac{\sqrt{\mu^n_{i}}}{\alpha_i(s)}\left(\sum_{k=1}^2\frac{\mu_k^n}{\alpha_k(s)}\right)^{-1/2}dt > d_{i}(x_i,c_{i}(t_n))\quad i=1,2.
\]
Then, by the standard limit process, we deduce the following inequalities:
\[
\int_{\Omega^+}^{r-\epsilon}\frac{\sqrt{\mu^*_{i}}}{\alpha_i(s)}\left(\sum_{k=1}^2\frac{\mu_k^*}{\alpha_k(s)}\right)^{-1/2}dt \geq d_{i}(x_{i},x_{i}^{*}),\quad i=1,2,
\] where $\mu_i^*$ is the limit (up to a subsequence) of $\{\mu_i^n\}$. Now observe that some of previous inequalities become strict if we replace $r-\epsilon$ by $r$. So, a small variation of $\mu_1^*$ and $\mu_2^*$ if necessary (concretely, if one of these constants is zero), provides
positive constants $\mu_{1}',\mu_{2}'>0$ satisfying
\[
\int_{\Omega^+}^{r}\frac{\sqrt{\mu'_{i}}}{\alpha_i(s)}\left(\sum_{k=1}^2\frac{\mu'_k}{\alpha_k(s)}\right)^{-1/2}dt > d_{i}(x_{i},x_{i}^{*}),\quad i=1,2.
\]
In particular, $(r,x_1,x_2)\in A_{(\Omega,x_1^{*},x_2^*)}$, and so, $\uparrow I^{-}(\gamma) \subset A_{(\Omega,x_1^{*},x_2^*)}$.

\smallskip

 For the inclusion to the left, assume that $(r,x_1,x_2) \in A_{(\Omega,x_1^*,x_2^*)}$. By the continuity of both, the integral with respect to the lower limit of integration and the distance function, and the convergence of $\gamma(t)=(t,c_1(t),c_2(t))$ to
$(\Omega,x_1^{*},x_2^*)$, we deduce that
\[
\int_{t}^{r}\frac{\sqrt{\mu_{i}}}{\alpha_i(s)}\left(\sum_{k=1}^2\frac{\mu_k}{\alpha_k(s)}\right)^{-1/2} > d_{i}(x_{i},c_{i}(t))\quad\hbox{for large $t$.}
\]
So, from Prop. \ref{c0}, $\gamma(t) \ll  (r,x_1,x_2)$ for  all $t$, which implies $(r,x_1,x_2) \in \uparrow I^{-}(\gamma)$.

\smallskip

For the last assertion, assume that $P$ is associated to $(\Omega^+,x^*_1,x^*_2)\in (a,b)\times M_1^C\times M_2^C$. From the first part of this proposition, $\uparrow P= I^+(\sigma)$, where $\sigma$ is a past-directed timelike curve converging to $(\Omega^+,x^*_1,x^*_2)$. So, $F=I^+(\sigma)$ is the unique maximal IF into the common future of $P$. Reasoning analogously we deduce that $P$ is the unique maximal IP into the common past of $F$. In conclusion, $P$ is $S$-related just with the indecomposable future set $F$, and vice versa.
\end{proof}



From this result it is clear that $\overline{V}$ is simple as a point set (see Defn. \ref{simpletop}). On the other hand, if we define
\[
\partial^{\B}V:=\hat{\partial}^{\B}V\cup\check{\partial}^{\B}V,
\]
the following identification is deduced:
\[
\overline{V}\setminus \partial^{\B} V\leftrightarrow (a,b) \times M_1^C\times M_2^C.
  \]
  In particular, $\partial V\setminus \partial^{\B} V$ can be identified with a cone with base $(M_1^C\times M_2^C)\setminus (M_1\times M_2)$. Moreover, if we assume that both $M_1^C,M_2^C$ are locally compact, Prop. \ref{topcurvasfinitas} ensures that previous bijection is a homeomorphism. Particularly, this proves that, given $(P,F)\in \overline{V}\setminus \partial^{\B}V$,
    \[
P\in \hat{L}(\{P_n\}) \iff F\in \check{L}(\{F_n\})
      \]
      for any sequence $\{(P_n,F_n)\}_n\in \overline{V}$. Hence, $\overline{V}\setminus \partial^{\B}V$ is also simple topologically.

  Finally, the following lemma ensures that the line over each point $(x_1^*,x_2^*)\in (M_1^C\times M_2^C)\setminus (M_1\times M_2)$ is timelike:
 \begin{lemma}\label{causalstructurenoinf}
     If $(P,F),(P',F')\in \partial V\setminus \partial^{\B}V$, with $(P,F)\equiv (\Omega,x^*_1,x^*_2), (P',F')\equiv (\Omega',x^*_1,x^*_2)$ in $(a,b)\times M_1^C\times M_2^C$, satisfy that $a<\Omega<\Omega'<b$ then $(P,F)\ll (P',F')$.
  \end{lemma}
  \begin{proof}
     Take $t=(\Omega+\Omega')/2$ and $\mu_1=\mu_2=1/2$. For $i=1,2$, consider $y_i$ close enough to $x_i^*$ so that
    \[
\left\{\begin{array}{l}
\displaystyle\Integral{t}{\Omega'}{\mu_{i}}{i}{\mu_{k}}>
             d_{i}(y_i,x_i^*) \\

\displaystyle\Integral{\Omega}{t}{\mu_{i}}{i}{\mu_{k}}>
             d_{i}(y_i,x_i^*),

\end{array}\right.
\quad i=1,2.
      \]
    From Prop. \ref{pastofcurve} (and its past analogous) we deduce that $(t,y_1,y_2)\in F\cap P'$, as desired.
  \end{proof}

  \smallskip

  The $S$-relation described in Prop. \ref{tip} implies that each pair $(P,F)\in\overline{V}$ is determined by any of its non-empty components, that is, $\overline{V}$ is simple as a point set. Even more, from Prop. \ref{topcurvasfinitas} and the definition of the chronological limit (see \eqref{eq:29} and \eqref{limcrono}), $\overline{V}$ is topologically simple as well (recall Defn. \ref{simpletop}); concretely, if $(P,F)\in \overline{V}$, $P\neq\emptyset$, and $\sigma=\{(P_n,F_n)\}_n\subset \overline{V}$, then $(P,F)\in L_{chr}(\sigma)$ if, and only if, $P\in \hat{L}_{chr}(\{P_n\}_n)$. Therefore, in order to determine the, pointwise and topological, structure of the (total) $c$-boundary, it suffices to study the partial boundaries. Consequently, we will describe $\overline{V}$ in two different ways, according to our convenience, namely:
  \[
\overline{V}= (a,b) \times M_1^C\times M_2^C\cup\hat{\partial}^{\B}V\cup\check{\partial}^{\B}V=\hat{V}\cup \check{\partial}^{\B}V=\hat{\partial}^{\B}V\cup \check{V}.
    \]
     Restricting conveniently, the open sets of $\overline{V}$ containing a pair $(P,F)$ can be viewed as: (i) open sets in $(a,b) \times M_1^C\times M_2^C$ if $P\neq\emptyset\neq F$, (ii) open sets in $\hat{V}$ if $F=\emptyset$ or (iii) open sets in $\check{V}$ if $P=\emptyset$.

\smallskip

It rests to determine the causal structure of $\overline{V}$. This is contained in the following result, which summarizes all the information about the (total) c-completion of doubly warped spacetimes:
  \begin{thm}\label{thm:main}
    Let $(V,g)$ be a {\multiwarped} spacetime as in \eqref{eq:1-aux}. Then, there exists a homeomorphism
    \[
\overline{V}\setminus \partial^{\B}V \leftrightarrow (a,b) \times M_1^C\times M_2^C,
      \]
    where each line $\{(t,x_1^*,x_2^*): t\in (a,b),\; (x_1^*,x_2^*)\in M_1^C\times M_2^C\}$ is timelike. Moreover:
     \begin{itemize}
      \item[(i)] If \eqref{eq:7} and \eqref{eqq:7} hold, then $\partial^{\B} V$ is homeomorphic to a couple of spacelike copies of $M_1^C\times M_2^C$. As consequence, we have the following homeomorphism:
        \begin{equation}
          \label{eq:18}
         \overline{V}\equiv [a,b]\times M_1^C\times M_2^C\quad\hbox{pointwise and topologically.}
        \end{equation}

        \item[(ii)] If \eqref{eq:7} and \eqref{eqq:9} hold, then $\partial^{\B} V$ has a copy of $M_1^C\times M_2^C$ for the future, with spatial causal character; and a copy of $M_1^C\times \left({\cal B}(M_2)\cup \{\infty\}\right)$ for the past.  This second set can be seen as a cone with base $M_1^C\times \partial_{\cal B}(M_2)$ generated by horismotic lines over each pair $(x_1^*,[b_{c_2}])$ ending at the point $(x_1^*,\infty)$. As consequence, we have the following homeomorphism
          \[
              \overline{V}\equiv\left\{\begin{array}{l} \hat{V}\cup \check{\partial}^{\B}V \leftrightarrow \left((a,b]\times M_1^C\times M_2^C\right) \cup \left(M_1^C\times \left({\cal B}(M_2)\cup \{\infty\}\right)\right) \\ \hat{\partial}^{\B} V\cup \check{V} \leftrightarrow \left(\{b\}\times M_1^C\times M_2^C\right) \cup \left(M_1^C\times \left(B(M_2)\cup \{\infty\} \right) \right).
            \end{array}\right.
            \]

          \item[(iii)] If \eqref{eq:9} and \eqref{eqq:7} hold we have a structure analogous to (ii), but interchanging the roles of future and past.

            \item[(iv)] If \eqref{eq:9} and \eqref{eqq:9} hold, then $\partial^{\B} V$ has two copies of the space $M_1^C\times \left({\cal B}(M_2)\cup \{\infty\}\right)$, one for the future and the other one for the past, formed by horismotic lines over each point $(x_1^*,[b_{c_2}])\in M_1^C\times \partial_{\cal B}(M_2)$ ending at the point $(x_1^*,\infty)$. As consequence,
              \[
                  \overline{V}\equiv \left\{\begin{array}{l} \hat{V}\cup \check{\partial}^{\B}V \leftrightarrow \left(M_1^C\times \left(B(M_2)\cup \{\infty\} \right)\right) \cup \left(M_1^C\times \left({\cal B}(M_2)\cup \{\infty\}\right)\right) \\
                  \hat{\partial}^{\B} V\cup \check{V} \leftrightarrow \left( M_1^C\times \left({\cal B}(M_2)\cup \{\infty\}\right)\right) \cup  \left(M_1^C\times \left(B(M_2)\cup \{\infty\}\right)\right).
                \end{array}\right.
            \]
      \end{itemize}
    \end{thm}
    \begin{proof}
As we have argued before, the first assertion about the point set topological and causal structure of $\overline{V}\setminus \partial^{\ncambios{b}}V$ is a direct consequence of \cambios{Props. \ref{structuraparcialsininfinito'}, \ref{topcurvasfinitas} (and its past analogous)}, \ref{tip} and Lemma \ref{causalstructurenoinf}. So, we will focus on the rest of assertions.

  \begin{itemize}
  \item[(i)] The point set and topological structure are straightforward from Thms. \ref{futurestructurefiniteconditions} and \ref{pfuturestructurefiniteconditions}. So, we only need to prove that $\partial^{\B}V=\hat{\partial}^{\B}V\cup \check{\partial}^{\B}V$ is spacelike. Take $(P,\emptyset),(P',\emptyset)\in \partial^{\B}V$ two different boundary points (for TIFs is completely analogous). By using the identification in \eqref{eq:18}, we can assume that $(P,\emptyset)\equiv (b,x_1^*,x_2^*)$ and $(P',\emptyset)\equiv (b,y_1^*,y_2^*)$ with $(x_1^*,x_2^*)\neq (y_1^*,y_2^*)$. From the proof of Prop. \ref{structuraparcialsininfinito} (recall also Rem. \ref{rem:1}) it follows both, $P\not \subset P'$ and $P'\not \subset P$, thus $(P,\emptyset)$ and $(P',\emptyset)$ are neither timelike nor lightlike related, i.e., they are spatially related.

  \item[(ii)] The point set and topological structure are deduced from Thm. \ref{futurestructurefiniteconditions} and Thm. \ref{pfuturecomploneinfinite}. For the causal structure, let us take two points $(P,\emptyset),(P',\emptyset)\in \partial^{\B}V$ over the same point $(x_1^*,[b_{c}])\in M_1^C \times \partial_{\cal B} M_2$. Hence, we can make the identifications $(P,\emptyset)\equiv (x_1^*,b_{c_1})$ and $(P',\emptyset)\equiv (x_1^*,b_{c_2})$ with $b_{c_1}-b_{c_2}={\cal K}$, ${\cal K}$ constant. If we assume that ${\cal K}>0$, then $b_{c_1}\geq b_{c_2}$, and so, $P'\subset P$ (recall Lemma \ref{lemma:aux5}), i.e., both points are lightlike related (the case with ${\cal K}<0$ is completely analogous).
  \end{itemize}

Finally, assertions (iii) and (iv) are easily deduced from (i) and (ii).

\end{proof}

\begin{rem} {\rm (1) Of course, the four cases considered in previous theorem do not cover all the possibilities compatible with the finiteness of at most one warping integral (since the finite warping integral may not be necessarily the last one). However, the structure of the c-completion for these additional cases are easily deducible from our approach, and can be considered an easy exercise for the reader.

(2) In order to simplify the exposition, we have considered along this paper multiwarped spacetimes with just two fibers. Nevertheless, the corresponding results for the general case of $n$ fibers can be easily deduced by the reader (see, for instance, Section \ref{sec:applications}).}
\end{rem}

\section{Some examples of interest}
\label{sec:applications}

In this section we are going to apply our results to compute the c-completion of some spacetimes of physical interest. Concretely, we will consider some Kasner models, the intermediate region of Reissner-Nordstr\"om and de Sitter models with (non necessarily compact) internal spaces.

%
    \subsection*{Kasner models}
{\em Generalized Kasner models} are multiwarped  spacetimes $(V,g)$ where $V=(0,\infty)\times \R^{n}$ and
\begin{equation}
  \label{eq:35}
g=-dt^2+t^{2p_1}dx_1^2+\dots +t^{2p_{n}}dx_{n}^2,\quad\hbox{$(p_1,\ldots,p_n)\in\R^n$}.
\end{equation}
These models are solutions to the vacuum Einstein equations if $(p_1,\dots,p_{n})\in \R^{n}$ belongs to the so-called {\em Kasner sphere}, i.e., if it satisfies
  \[
\sum_{i=1}^{n}p_i=1=\sum_{i=1}^{n}p^2_i.
\]
Even if this condition does not fall under the hypotheses of our results, this does not cover all the cases of interest, and so, we are not going to assume it.

As far as we know, the c-boundary of these models can be faced in two different ways. On the one hand, by using Harris' result (Thm. \ref{thm:harris}); taking into account that the fibers are complete, this result gives a full description of the future c-boundary when $p_i>1$ for all $i$, and provides some partial information in the other cases. On the other hand, these models have been studied
 by Garc\'ia-Parrado and Senovilla in \cite{GS03} by using the isocausal relation. They essentially prove that, depending on the values of the constants $p_1,\ldots,p_n$, the corresponding Kasner model is isocausal to a particular Robertson-Walker model whose c-boundary is well-known. This may be useful, since, although the c-boundary of isocausal spacetimes may be different (see \cite{0264-9381-28-17-175016}), they can share some qualitative properties (see \cite{FHSIso2}).

 \smallskip

Of course, Thm. \ref{futurestructurefiniteconditions} parallels Harris' result for Kasner models when $p_i>1$ for all $i$. However, now we can go a step further and give a complete description of the c-boundary when
\[
p_i>1\;\;\hbox{for $1\leq i\leq k$,}\quad p_i=q\;\;\hbox{for $k+1\leq i\leq n$}\quad\hbox{and}\quad \int_1^{\infty}\frac{1}{t^{q}}dt=\infty.
\]
In this case we can write
\[
V=(0,\infty)\times\R^k\times\R^{n-k},\qquad g=-dt^2+\sum_{i=1}^{k}t^{2p_i}dx_i^2+t^{2q}\left(\sum_{i=k+1}^ndx_i^2 \right),
\]
In particular,
\[
\int_{1}^{\infty}\frac{dt}{t^{p_i}}<\infty,\;\; i=1,\ldots,k,\qquad\int_{1}^{\infty}\frac{dt}{t^{q}}=\infty.
\]
Therefore, the spacetime falls under the hypotheses of (the obvious multiwarped version of) Thm. \ref{futurecomploneinfinite} (essentially, with $M_1=\R^k$ and $M_2=\R^{n-k}$), which provides
the following homeomophism:
          \[
\hat{V}\;\leftrightarrow\; \left((0,\infty)\times \R^n\right)\cup \left(\R^k\times \left({\cal B}(\R^{n-k})\cup \{\infty\}\right)\right).
            \]
            So, taking into account that (see, for instance, \cite[Section 5.1]{H2})
            \[
{\cal B}(\R^{n-k})\equiv \R\times \mathbb{S}^{n-k-1},
              \]
              we immediately deduce that
              \[
\hat{\partial} V\leftrightarrow \R^k \times \left(\left(\R\times \mathbb{S}^{n-k-1}\right)\cup \{\infty\} \right).
                \]

\subsection*{The intermediate  Reissner-Nordstr\"om}
The Reissner-Nordstr\"om model is a spacetime $(V,g)$, where $V=\R\times \R\times \mathbb{S}^2$ and
\[
g=-\left(1-\frac{2m}{r}+\frac{q^{2}}{r^{2}}\right)dt^{2}+\left(1-\frac{2m}{r}+\frac{q^{2}}{r^{2}}\right)^{-1}dr^{2}+r^{2}(d\theta^{2}+sin^{2}\theta d\phi^{2}).
\]
 This metric degenerates at the zeros of the function $f(r)=(1-2m/r+q^2/r^2)$, which depend on the parameters $m$ (mass) and $q$ (charge). For our purposes we will require that $q\leq m$, which ensures the zeros $r^{\pm}=m\left( 1\pm\sqrt{1-q^2/m^2}\right)$ for $f$. The {\em intermediate region} of the Reissner-Nordstr\"om  is the spacetime $(V_I,g)$, where $V_I=\R\times (r^-,r^+)\times \mathbb{S}^2$.

Taking into account that $f(r)<0$ on $(r^-,r^+)$, the metric $g$ can be rewritten on $V_I$ as
\begin{equation}
  \label{eq:37}
g= -f(r)dt^2 + \frac{1}{f(r)}dr^2 + r^2 d\sigma^2=-d\tau^2 + r(\tau)^2d\sigma^2-F(\tau)dt^2,
\end{equation}
where
\[
d\tau:=-\frac{dr}{\sqrt{f(r)}}=\frac{dr}{\sqrt{-1+2m/r-q^2/r^2}} \qquad \hbox{and}\qquad F(\tau)=f(r(\tau)).
\]
Note that $\tau$ ranges in a finite interval $(a,b)$, and so, $(V_I,g)$ clearly corresponds with the standard form of a doubly warped spacetime where $V_I=(a,b)\times\mathbb{S}^2\times \R$.
In order to proceed with the analysis of the c-completion of $(V_I,g)$, we need to distinguish two cases: $q\neq 0$ and $q=0$.\footnote{Since the Penrose's diagram of Reissner-Nordstr\"om is well-known (see, for instance, \cite{hawking1975large}), the c-completion of $(V_I,g)$ can be also studied by applying \cite[Thm. 4.32]{FHSFinalDef}.}
%
%

\subsubsection*{Intermediate Reissner-Nordstr\"om with charge, $q\neq 0$.}

In this case, the warping integrals satisfy, for $a<c<b$,
 \begin{align}
   \int_{a}^{c}\frac{1}{\sqrt{\alpha_1(\tau)}}d\tau = \int_{r^-}^{r(c)}\frac{1}{r\sqrt{-1+\frac{2m}{r}-\frac{q^2}{r^2}}}dr<\infty\label{eq:38a} \\
   \int_{c}^{b}\frac{1}{\sqrt{\alpha_1(\tau)}}d\tau=\int^{r^+}_{r(c)}\frac{1}{r\sqrt{-1+\frac{2m}{r}-\frac{q^2}{r^2}}}dr<\infty\label{eq:38b}
  \end{align}
and
  \begin{align}
    \int_{a}^{c}\frac{1}{\sqrt{\alpha_2(\tau)}}d\tau=\int_{r^-}^{r(c)}\frac{1}{-1+\frac{2m}{r}-\frac{q^2}{r^2}}dr=\infty \label{eq:38}\\
    \int_{c}^{b}\frac{1}{\sqrt{\alpha_2(\tau)}}d\tau=\int^{r^+}_{r(c)}\frac{1}{-1+\frac{2m}{r}-\frac{q^2}{r^2}}dr=\infty.\label{eq:38c}
  \end{align}
  So, from Thm. \ref{thm:main} (iv) (with $M_1=\mathbb{S}^2$ and $M_2=\R$), we deduce the homeomorphisms
  \[
    \begin{array}{c}
      \overline{V}\leftrightarrow \left((a,b)\times \mathbb{S}^2\times \R\right) \cup (\mathbb{S}^2\times \left(\left(\R\times \{z^-, z^+\}\right)\cup \{i^+\} \right))\cup (\mathbb{S}^2\times \left(\left(\R\times \{z^-, z^+\}\right)\cup \{i^-\}\right)),\\
      \\
\partial V\leftrightarrow (\mathbb{S}^2\times \left(\left(\R\times \{z^-, z^+\}\right)\cup \{i^+\} \right))\cup (\mathbb{S}^2\times \left(\left(\R\times \{z^-, z^+\}\right)\cup \{i^-\}\right)),
    \end{array}
    \]
    where we have used that ${\cal B}(\R)\equiv \R\times \{z^-,z^+\}$, being $z^-$ and $z^+$ the two asymptotic directions (left and right) of $\R$.

  \subsubsection*{Interior Schwarzschild, $q=0$.}

When $q=0$, $f(r)$ has only one zero, we can identify $(r^-,r^+)\equiv (0,2M)$, and the intermediate region of Reissner-Nordstr\"om coincides with the interior region of Schwarzschild. In this case, the warping integrals (\ref{eq:38a}), (\ref{eq:38b}) and (\ref{eq:38c}) still hold, but \eqref{eq:38} transforms into
  \[
\int_{a}^{c}\frac{1}{\sqrt{\alpha_2(\tau)}}d\tau = \int_{0}^{r(c)}\frac{1}{-1+\frac{2m}{r}}dr<\infty.
  \]
  So, from Thm. \ref{thm:main} (iii), we deduce the homeomorphism
  \[
\overline{V}\leftrightarrow \left([a,b)\times \mathbb{S}^2\times \R\right) \cup \left( \mathbb{S}^2\times \left(\R\times \{z^-,z^+\} \right)\right)
    \]
  and thus,\footnote{The usual time-orientation on Reissner-Nordstr\"om makes the vector field $\partial_r$ past-directed in the intermediate region. So, in formula (\ref{d}), the roles of the future and past c-boundaries are interchanged with respect to the (a priori) expected ones.}
 \begin{equation}\label{d}
  \partial V\equiv \hat{\partial} V\cup \check{\partial} V \leftrightarrow  \left(\{a\}\times \mathbb{S}^2\times \R   \right) \cup \left(\mathbb{S}^2\times \left( \left(\R\times \{z^-,z^+\}\right)\cup \{i^+\} \right)\right).
  \end{equation}

\subsection*{De Sitter models with (non-necessarily compact) internal spaces}

Motivated by the relevance for the problem of the dS/CFT correspondence, finally we study the c-boundary of warped products of de Sitter models with general Riemannian manifolds.

Recall that {\em de Sitter spacetime} can be seen as a Robertson-Walker spacetime $(M,g_{M})$, where
\[
M=\R\times \mathbb{S}^l,\qquad g_{M}=-dt^2 + cosh(t)^2 g_{\mathbb{S}^l}.
  \]
Consider the doubly warped spacetime $(V,g)$ obtained as the product of de Sitter space $(M,g_{M})$ and a Riemannian manifold $(F,g_{F})$, i.e.,
  \[V=\R\times \mathbb{S}^l\times F,\qquad g=-dt^2+cosh^2(t)g_{\mathbb{S}^{l}}+g_{F}.
    \]
  The first warping function $\alpha_1(t)=\cosh(t)^2$ satisfies the finite integral conditions for both, the future and the past directions, meanwhile the second one $\alpha_2(t)\equiv 1$ does not. Therefore, from Thm. \ref{thm:main} (iv) (with $M_1=\mathbb{S}^l$ and $M_2=F$), we deduce the following homeomorphism for the c-boundary of $(V,g)$:
\[
\partial {V}\equiv \hat{\partial} V \cup \check{\partial} V \leftrightarrow  \left(\mathbb{S}^l\times \left({\cal B}(F)\cup \{i^+\}\right) \right)\,  \cup \, \left(\mathbb{S}^l\times \left({\cal B}(F)\cup \{i^-\}\right) \right).
\]
  In particular, if $(F,g_{F})$ is compact, then ${\cal B}(F)$ is empty, and the c-boundary becomes (compare with the last assertion on Thm. \ref{thm:harris}):
\[
  \partial V\leftrightarrow \left(\mathbb{S}^l\times \{i^+\}) \right)\,  \cup \, \left(\mathbb{S}^l\times \{i^-\} \right).
  \]